%% file: main.tex
\documentclass[twoside,leqno,twocolumn]{article}
\usepackage[letterpaper]{geometry}
\usepackage[T1]{fontenc}

\usepackage{thm-restate}

\usepackage{ltexpprt}
\usepackage{hyperref}
\usepackage{xspace}

\usepackage{graphicx}
\usepackage{amsmath}

\usepackage[capitalise, noabbrev]{cleveref}
\crefname{algocf}{Algorithm}{Algorithms}
\Crefname{algocf}{Algorithm}{Algorithms}
\crefname{@theorem}{Theorem}{Theorems}
\Crefname{@theorem}{Theorem}{Theorems}

\usepackage{amsfonts}
\usepackage{amssymb}
\usepackage{stmaryrd}
\usepackage{autobreak}
\usepackage{placeins}
\usepackage{cases}

\usepackage[binary-units=true]{siunitx}
\usepackage[subrefformat=simple,labelformat=simple]{subcaption}

\usepackage[section]{algorithm}
\usepackage{algorithmicx}
\usepackage{algpseudocode}

\algdef{SE}[DOWHILE]{Do}{doWhile}{\algorithmicdo}[1]{\algorithmicwhile\ #1}%

\crefname{claim}{Claim}{Claims}

\newcommand{\old}[1]{{}}

\newcommand{\atsp}{\alpha_{\mbox{\scriptsize{TSP}}}}

\newtheorem{problem}{Problem}

\input{helper/commands.tex}
\begin{document}

\newcommand\relatedversion{}

\title{\Large A Closer Cut:\\ Computing Near-Optimal Lawn Mowing Tours\relatedversion}
\author{Sándor P.~Fekete\thanks{Department of Computer Science, TU Braunschweig, Braunschweig, Germany}
\and Dominik Krupke\footnotemark[1]
\and Michael Perk\footnotemark[1]
\and Christian Rieck\footnotemark[1]
\and Christian Scheffer\thanks{Faculty of Electrical Engineering and Computer Science, Bochum University of Applied Sciences, Bochum, Germany}
}
\date{}

\maketitle


\fancyfoot[R]{\scriptsize{Copyright \textcopyright\ 2022\\
Copyright for this paper is retained by authors.}}

\begin{abstract}
For a given polygonal region $\Pol$, the Lawn Mowing Problem (LMP) asks for a shortest
tour $\T$ that gets within Euclidean distance 1 of every point in $\Pol$; this is
equivalent to computing a shortest tour for a unit-disk cutter $C$ that covers all
of $\Pol$. As a geometric optimization problem of natural
 practical and theoretical importance, the LMP generalizes and combines several notoriously
difficult problems, including minimum covering by disks,
the Traveling Salesman Problem with neighborhoods~(TSPN), and the $\exists\R$-complete
Art Gallery Problem (AGP).
So far, there have only been theoretical approximation algorithms with worst-case
bounds of $2\sqrt{3}\atsp\approx 3.46 \atsp$,
where $\atsp$ is the approximation factor for the geometric TSP\@. Here,
$\atsp=1+\varepsilon$ is theoretically possible by using one of the famous
geometric approximation schemes; however, these methods are not practically applicable
for concrete instances. Moreover, there have not been any exact methods for the LMP
that compute provably near-optimal solutions for instances of interesting size, owing to
the combination of geometric difficulties, such as a succinct characterization of optimal solutions,
as well as the lack of useful lower bounds that provide practically small performance gaps.

In this paper, we conduct the first study of the Lawn Mowing Problem with a focus
on practical computation of near-optimal solutions. To this end, we provide
new theoretical insights: Optimal solutions are polygonal paths with a bounded
number of vertices, i.e., they do not have any curved pieces, allowing a restriction
to straight-line solutions;  on the other hand, there can be relatively simple instances
for which optimal solutions require a large class of irrational coordinates.
On the practical side, we present a primal-dual approach with provable convergence properties
based on solving a special case of the TSPN restricted to \emph{witness sets}.
In each iteration, this establishes both a valid solution and
a valid lower bound, and thereby a bound on the remaining optimality gap.
As we demonstrate in an extensive computational study, this allows us to achieve
provably optimal and near-optimal solutions for a large spectrum of benchmark
instances with up to \num{2000} vertices.
\end{abstract}

\input{01-introduction.tex}
\input{02-straight-line-tours.tex}

\input{03-irrational-tour-vertices.tex}
\input{04-cetsp-approach.tex}
\input{05-experiments.tex}
\input{06-conclusion.tex}

\FloatBarrier
\cleardoublepage
\bibliography{main}
\clearpage
\appendix

\input{04C-cetsp-approach}
\input{05A-additional-figures}
\end{document}

%% file: helper/commands.tex
\newcommand{\Pol}{P}
\newcommand{\T}{T}
\newcommand{\Topt}{T^*}
\newcommand{\R}{\mathbb{R}}

\newcommand{\cupdot}{\mathbin{\mathaccent\cdot\cup}}

\DeclareMathOperator*{\argmin}{arg\,min}

\newcommand{\ceil}[1]{\left\lceil #1 \right\rceil}

%% file: 01-introduction.tex
\section{Introduction}

There are many facets of theoretical and practical difficulty of geometric optimization problems.
On the theoretical side, the classic \emph{Traveling Salesman Problem}~(TSP) is NP-hard, making it unlikely that there
is a polynomial-time algorithm that produces provably optimal solutions. Moreover, it
is unknown how to efficiently evaluate a sum of square roots, so it
is unclear whether the TSP for a set of points in the plane with Euclidean distances even belongs
to NP\@. For the famous \emph{Art Gallery Problem} (AGP) of finding a minimum number of
guards to cover a simple polygon based on visibility, membership in NP is indeed unlikely,
as it belongs to the class $\exists\R$. Furthermore, problems of {optimal covering}
are also known to be prohibitively difficult from a \emph{practical} perspective: Even the
largest square that can be covered by $n$ unit disks has only been established up to $n=7$.

We consider a geometric optimization problem that generalizes
and combines these challenges. In the \emph{Lawn Mowing Problem} (LMP),
we are given a (not necessarily simple or even connected) polygonal region $\Pol$
and a disk cutter $C$ of radius $r$; the task is to find a closed roundtrip of minimum
Euclidean length such that the cutter ``mows'' all of $\Pol$, i.e., a~shortest tour that moves
the center of $C$ within distance $r$ from every point in~$\Pol$. The LMP naturally
occurs in a wide spectrum of practical applications, such as robotics,
manufacturing, farming, quality control and image processing, so it is of both
theoretical and practical importance.

\begin{figure*}
\centering
\begin{subfigure}{.19\textwidth}
  \centering
  \includegraphics[width=\textwidth]{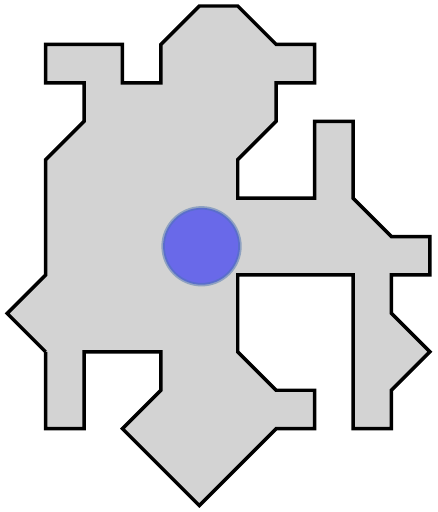}
  \caption{}
\end{subfigure}\hfill
\begin{subfigure}{.19\textwidth}
  \centering
  \includegraphics[width=\textwidth]{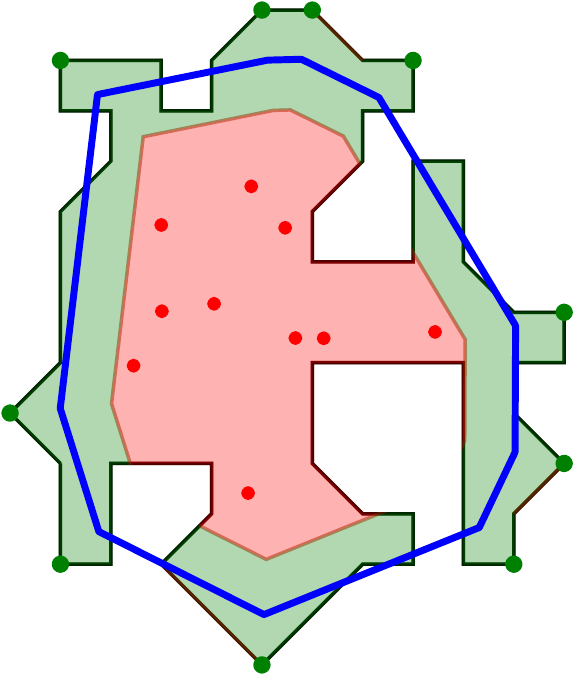}
  \caption{}
\end{subfigure}\hfill
\begin{subfigure}{.19\textwidth}
  \centering
  \includegraphics[width=\textwidth]{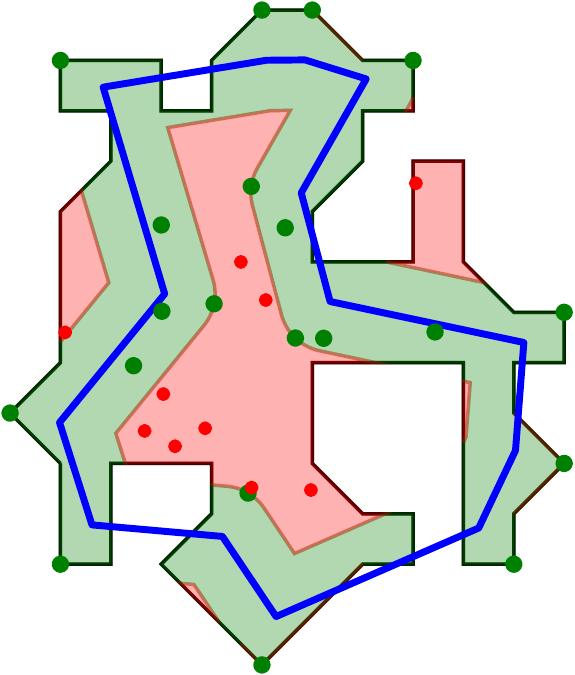}
  \caption{}
\end{subfigure}\hfill
\begin{subfigure}{.19\textwidth}
  \centering
  \includegraphics[width=\textwidth]{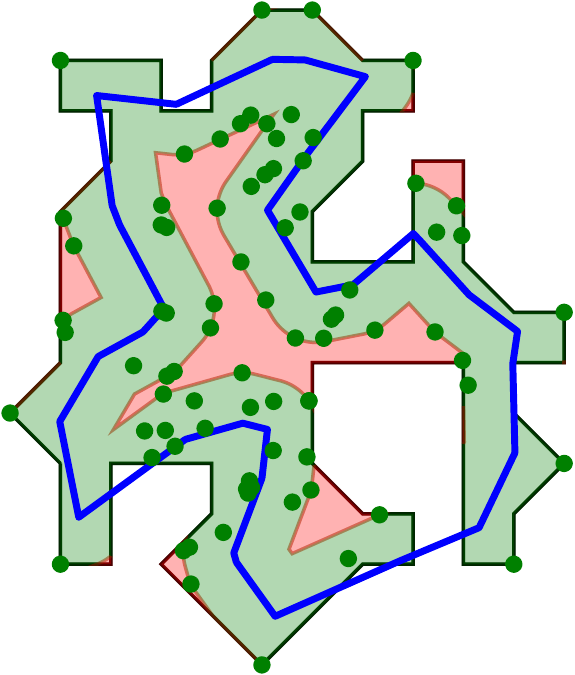}
  \caption{}
\end{subfigure}\hfill
\begin{subfigure}{.19\textwidth}
  \centering
  \includegraphics[width=\textwidth]{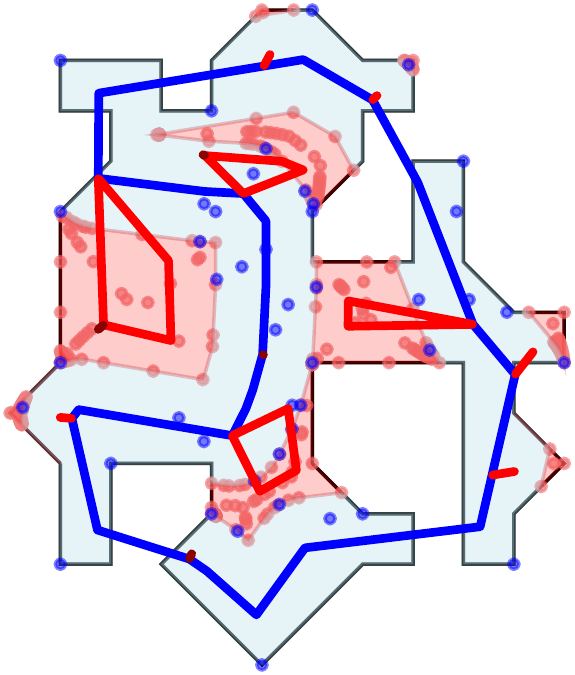}
  \caption{}
\end{subfigure}
  \caption{A sequence of iterations in the primal-dual scheme. \textbf{(a)} An example instance. The blue circle shows the tool size.
     \textbf{(b)} A lower bound of \num{32.4} with coverage of \SI{51}{\percent}, arising in iteration 0 from an initial witness set.
     \textbf{(c)}~A lower bound of \num{36.56} with coverage of \SI{70}{\percent}, arising in iteration 1 from an enhanced witness set.
     \textbf{(d)} A~lower bound of \num{42.37} with coverage of \SI{84}{\percent}, arising in iteration~7.
     \textbf{(e)} An upper bound of \num{65.35} with full coverage, achieved in iteration 3.
  \label{fig:example}}
\end{figure*}

Given that the LMP combines the Euclidean TSP, the AGP, and covering by disks, it is not
surprising that it is also both theoretically and practically difficult: It is NP-hard
(as it generalizes the TSP), membership in NP is unclear (as it involves
evaluating Euclidean distances); it is also extremely hard from a practical perspective
(as it comprises covering by disks). In fact, the only known positive algorithmic result
is an approximation algorithm with a factor of
$2\sqrt{3}\atsp\approx 3.46 \atsp$~\cite{Arkin2000},
where $\atsp$ is the factor for the geometric TSP\@. So far,
no results aiming at methods with \emph{practically} good performance are known, in part
because of the difficulties of (I) characterizing optimal solutions (due
to the continuous nature of the LMP) and of (II) providing
tight lower bounds: Neither the area of $P$ nor its diameter can provide such bounds,
as replacing $P$ by a dense subset of points with zero area does not change the
length of an optimal tour, which can be much longer than the diameter of $P$.

\subsection{Our contribution}
We provide a number of theoretical and practical results for the Lawn Mowing Problem.
\begin{itemize}
\item We establish a characterization of optimal lawn mowing tours by proving
that an optimal tour for a polygonal region $\Pol$ consists of line segments
between a finite set of vertices. This allows a
focus on polygonal solutions, and the ensuing primal-dual scheme.
\item On the other hand, we show that even relatively simple regions $\Pol$
can require LMP solutions with a wide range of irrational vertices. While this does not establish
$\exists\R$-hardness, it gives some indication of the underlying difficulty,
as it did for the AGP\@.
\item We establish a primal-dual algorithm for the LMP by iteratively covering an
expanding \emph{witness set} of finitely many points in $\Pol$. In each iteration,
computing a lower bound involves solving a special case of a TSP instance with neighborhoods,
the \emph{Close-Enough TSP} (CETSP) to provable optimality; for an upper bound, this
is enhanced to provide full coverage. In each iteration, this establishes both a valid solution and
a valid lower bound, and thereby a bound on the remaining optimality gap.
(See \cref{fig:example} for an illustration.)
\item We prove that this discretization method leads to provably good results:
both the uncovered area and the maximum distance of points
from the region swept by our lower-bound tours converge to zero as we enhance the
witness set.
\item We present a comprehensive study to demonstrate the practical usefulness
of our methods, based on a wide spectrum of benchmark instances with up to 2000 vertices.
The outcomes include provably optimal solutions, limited optimality gaps, and improved lower bounds.
\end{itemize}


\subsection{Related work}
There is a wide range of practical applications for the LMP, including manufacturing~\cite{Arkin2000ZigZag,Held1991,Held},
cleaning~\cite{bormann2015new}, robotic coverage~\cite{cabreira2019survey,choset2001coverage,galceran2013survey,jensen2020near},
inspection~\cite{englot2012sampling}, CAD~\cite{elber1999offsets}, farming~\cite{bahnemannrevisiting,choset1998coverage,oksanen2009coverage}
and pest control~\cite{Becker}.
In Computational Geometry, the Lawn Mowing Problem was first introduced by Arkin et al.~\cite{arkin1993lawnmower},
who later gave the currently best approximation algorithm with a performance guarantee of
$2\sqrt{3}\atsp\approx 3.46 \atsp$~\cite{Arkin2000}, where $\atsp$
is the performance guarantee for an approximation algorithm for the TSP\@; while $\atsp$
may be $(1+\varepsilon)$ in theory, based on the methods of
Arora~\cite{arora1998polynomial} or Mitchell~\cite{mitchell1999guillotine}, neither approach is practically useful. A
variant (in which cost was incurred both for traveling and covering) was considered by Fekete et al.~\cite{fms-mctc-10},
who gave a 4-approximation for the special case of polyominoes. Closely related
is the TSP with neighborhoods (TSPN), for which it suffices to visit a neighborhood
for each of a given set of discrete vertices; this was first considered by
Arkin and Hassin~\cite{arkin1994approximation}, in a graph setting by
Gendreau et al.~\cite{gendreau1997covering}, for ``fat'' neighborhoods by
Mitchell~\cite{Mitchell07} and heuristically by Yuan and Zhang~\cite{yuan2017towards}.
A particularly relevant special case of the TSPN is the
\emph{Close-Enough TSP}~(CETSP), in which it suffices to get within a Euclidean
distance of $r$, i.e., for which the neighborhood is an $r$-disk. Dumitrescu and T\'oth~\cite{dumitrescu2017constant}
gave an $O(1)$~approximation; they also provide a broad overview of other theoretical results for the TSPN\@.
Practical methods were considered by Mennell~\cite{mennell2009heuristics}, Behdani and Smith~\cite{behdani2014integer},
and Coutinho et al.~\cite{coutinho2016branch}.

The Art Gallery Problem (AGP) is connected to our work by a combination of theoretical and practical
issues. Also a problem of optimal geometric covering, the AGP has to deal with the theoretical
difficulties of possibly irrational coordinates, as shown by Abrahamsen~et~al.~\cite{abrahamsen2017irrational}
(answering an open problem by Fekete~\cite{agarwal2011computational}); subsequently, this result
served as a stepping stone towards a proof of $\exists\R$-completeness~\cite{AbrahamsenAM22}.
On the practical side, powerful methods (e.g., by Baumgartner et al.~\cite{baumgartner2010exact},
Kr\"oller et al.~\cite{kroller2012exact}, or de Rezende et al.~\cite{rezende2016})
for computing a good set of guards for a polygonal region $\Pol$ are based on
finding solutions for discrete \emph{witness sets} within $\Pol$,
leading to a primal-dual optimization method;
see~\cite{bdd-pgpc-13} for an animated multimedia description. This approach
is closely related to our primal-dual method for the LMP\@.

Optimally covering even relatively simple regions by a set of $n$ unit disks has received
a considerable attention, but is excruciatingly difficult.
For covering rectangles by $n$ unit disks, Heppes and Mellissen~\cite{heppes1997covering}
gave optimal solutions for~$n\leq 5$; Melissen and Schuur~\cite{melissen2000covering} extended this
for $n=6,7$.  See the website by Friedman~\cite{friedman1014} for illustrations of
the best known solutions (only some of which are proven to be optimal) for $n\leq 12$.
As early as 1915, Neville~\cite{neville1915solution} computed the optimal arrangement for covering a disk by five unit disks,
but reported a wrong optimal value; much later, Bezdek~\cite{bezdek1979korok,bezdek1984einige} gave the correct value for $n=5,6$.
As recently as 2005, Fejes T\'{o}th~\cite{toth2005thinnest} established optimal values for $n=8,9,10$.
The question of incomplete coverings was raised
in 2008 by Connelly;
Szalkai~\cite{szalkai2016optimal} gave an optimal solution for $n=3$.
Progress on covering by (not necessarily equal) disks has been achieved by Fekete et al.~\cite{2020-Covering_SoCG,75-Covervideo_SoCG}.


\subsection{Preliminaries}\label{sec:preliminaries}

A~(simple)~\emph{polygon}~$\Pol$ is a (non-self-intersecting) shape in the
plane, bounded by a finite number $n$ of line segments. The \emph{boundary} of
a polygon~$\Pol$ is denoted by $\partial \Pol$. A~\emph{tour} is a closed
continuous curve $T: [0,1] \rightarrow \mathbb{R}^2$ with $T(0) = T(1)$;
we denote the (Euclidean) length of $T$ by $\ell(T)$.
The~\emph{cutter}~$C$ is a disk of radius $r$, centered in its midpoint. The
\emph{Minkowski sum} of two sets $A, B\subset \mathbb{R}^2$ is the set $A\oplus
B = \{a+b\ |\  a\in A,\,b\in B\}$. The \emph{coverage} of a tour $T$ with the
disk cutter $C$ is $T \oplus C$. The coverage of a point $p \in T$ is $\{
p \} \oplus C$. A \emph{lawn mowing tour}~$T$ of a polygon $P$ with a cutter
$C$ is a tour whose coverage contains $P$. An \emph{optimal} lawn mowing tour
$\Topt$ is a lawn mowing tour of shortest length. For a discrete set of points
$P$, a tour $T$ \emph{traverses} $P$ if $P \subset T$, i.e., for each $p \in P$
there is a $t \in [0,1]$ such that $T(t) = p$.


%% file: 02-straight-line-tours.tex
\begin{figure*}[ht]
	\centering
	\includegraphics[width=0.8\linewidth]{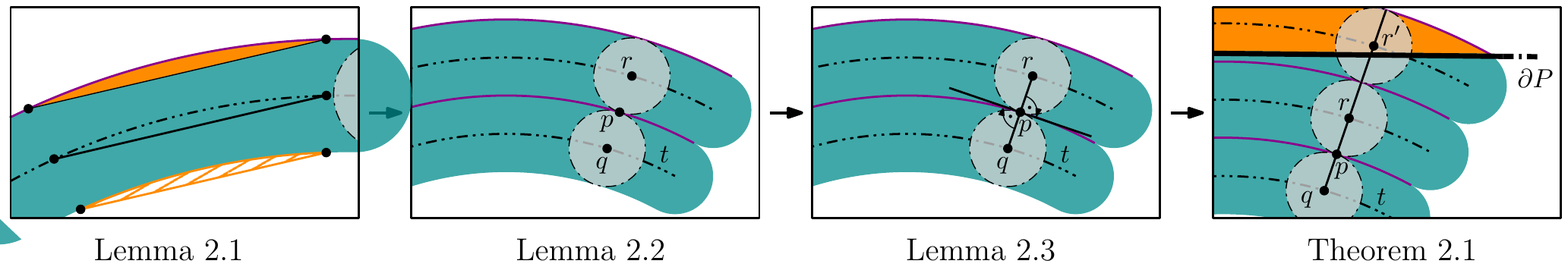}
	\caption{A symbolic overview of our proof that optimal tours only contain line segments.}
	\label{fig:straight-outline}
\end{figure*}

\section{Optimal tours are straight}\label{sec:straight-line-tours}

In general, a lawn mowing tour may consist of line segments and curved arcs. We show that any optimal tour mowing a polygon can only contain line segments.

\begin{restatable}{theorem}{straightlinetours}\label{theorem:straight-line-tours-circles}
	For every polygon $\Pol$ and a circular cutter
	, an optimal lawn mowing tour~$\T$ exclusively consists of a set of line segments.
\end{restatable}

%

\subsection{Outline of the proof}

Because our proof of \cref{theorem:straight-line-tours-circles} is rather technical, we formulate various lemmata within its context, see~\cref{fig:straight-outline}.
By~means of these lemmata, the proof is partitioned as follows.
\begin{description}
	\item[\cref{lemma:straight-line-shorten-tours}.] The key lemma for the proof of \cref{theorem:straight-line-tours-circles} is that given an optimal tour, the coverage of a curved arc cannot contain a region that is covered more than once.
	\item[\cref{lemma:two-closest-cutter-points}.] Based on~\cref{lemma:straight-line-shorten-tours}, we show that each point $p$ from the boundary of an arc $t$ has two closest points $q,r$ from the optimal tour.
	\item[\cref{lemma:orthogonal}.] We show that the two closest points $q,r$ to $p$ lie in a line with $p$; 
	we call $r$ the \emph{successor} of~$q$.
	\item[\cref{theorem:straight-line-tours-circles}.] Alternately
applying~\cref{lemma:two-closest-cutter-points,,lemma:orthogonal} leads to a sequence of successors on a common
line, 
with a final successor $r'$ outside the polygon. This
implies that the final successor lies on an arc that covers a region
outside of the polygon, leading to a contradiction by \cref{lemma:straight-line-shorten-tours}.
\end{description}

\subsection{Preliminaries for the proof}
A~\emph{(curved) arc} is the image of a smooth function $t:[0,1] \rightarrow
\mathbb{R}^2$ that is either strongly convex or strongly concave,
see~\cref{fig:conic-curve-bounded}. A~\emph{segment} is the image of a linear
function $t:[0,1] \rightarrow \mathbb{R}^2$. A~\emph{component} $t$ is an arc
or a segment. We call $t(0)$ and $t(1)$ the \emph{start} and \emph{end point}
of~$t$. Any tour $T: [0,1] \rightarrow \mathbb{R}^2$ can be partitioned into a
sequence of components $t_0,\dots,t_{k-1}$, such that the end point of $t_i$
equals the start point of $t_{i+1\bmod k}$ for $i \in \{ 0,\dots,k-1 \}$.

\begin{figure}[ht]
	\centering
	\includegraphics[width=0.6\columnwidth]{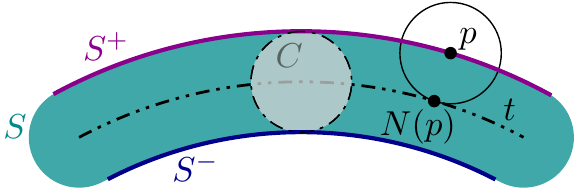}
	\caption{The coverage strip $S(t)$ (or $S$ in short) that results from traversing $t$ with the midpoint of the circular cutter $C$. $S^+$~and~$S^-$ denote the convex, and concave side, respectively. $p$ is a point in the interior of~$S^+$, and $N(p)$ is its closest point on $t$.}
	\label{fig:conic-curve-bounded}
\end{figure}

%

For any component $t$, the \emph{coverage strip} $S(t)$ is the region derived as the Minkowski sum $C \oplus t$. The~boundary $\partial S(t)$ can be partitioned into four curves as follows, see~\cref{fig:conic-curve-bounded}: For each point $p$ on the boundary, we consider its closest point $N(p)$ on $t$ regarding the Euclidean metric. 
The closure of all points $p$ on the boundary such that $N(p)$ is neither the
start nor the endpoint of $t$ is the union of two way-connected components. In
particular, when following the boundary of $S(t)$ in clockwise orientation,
the ``convex'' component (denoted by~$S^+(t)$) has positive curvature,
while the ``concave'' component (denoted by~$S^-(t)$) has
negative curvature; see~\cref{fig:conic-curve-bounded}.
If it is clear from the context, we denote by $S^+ := S^+(t)$, $S^- := S^-(t)$, and $S :=
S(t)$.

\subsection{Details of the proof}

A simple observation is that any tour can be partitioned into a set of line segments and curved arcs.
To show that the optimality of a tour excludes the existence of a curved arc, we consider
for the sake of contradiction
an optimal lawn mowing tour~$\T:=\{t_0,\dots, t_{k-1}\}$ for a circular cutter $C$ with at least one curved arc~$t := t_i$
.
Note that all our arguments are independent from the cutter's radius.

The key lemma for the proof 
is that ``shortcutting'' curved arcs by line segments while maintaining the coverage of the entire tour reduces its length.

\begin{lemma}\label{lemma:straight-line-shorten-tours}
Let $ab$ be a segment inside of the coverage strip $S$ of a curved arc $t$, with its end points $a,b$ on the convex side $S^+$ of $S$. Let $A$ be the region bounded by $ab$ and $S^+$. If~$A$ does not have to be covered by $t$ (because it is already
covered otherwise), then $T$ is not optimal.
\end{lemma}

\begin{proof}
Let $N(a), N(b) \in t$ be the closest points to $a, b$, respectively. Replacing the part of~$t$ between $N(a)$ and $N(b)$ by the segment $N(a)N(b)$ results 
in a shorter tour, see \cref{fig:cutter-strip-cropping-a}. This tour still covers the polygon, because the arc-segment-arc sequence $(t_1,N(a)N(b),t_2)$ covers $S$, except for the region $A$ that by assumption does not have to be covered by the sequence. Thus the tour containing the original arc~$t$ was not optimal.
  \begin{figure}[ht]
	\centering
	\begin{subfigure}[t]{0.49\linewidth}
		\centering
		\includegraphics[page=1,width=.9\linewidth]{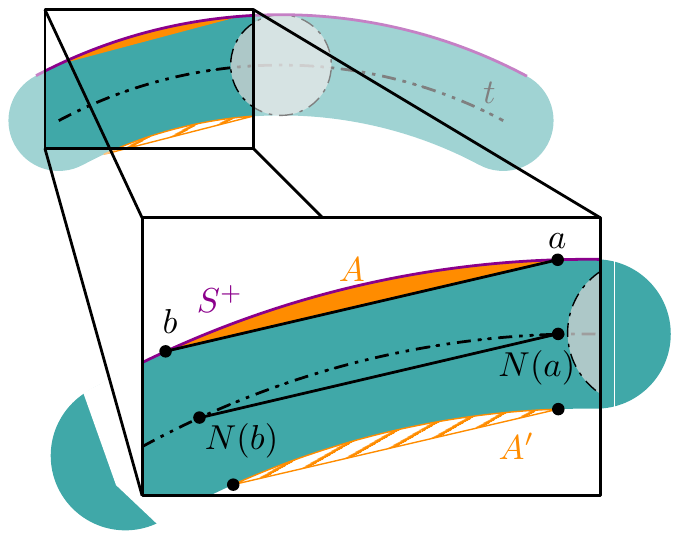}
		\caption{A covered region $A$.}
		\label{fig:cutter-strip-cropping-a}
	\end{subfigure}\hfill
	\begin{subfigure}[t]{0.49\linewidth}
		\centering
		\includegraphics[page=2,width=.9\linewidth]{figures/coverage-extra-region-new.pdf}
		\caption{Covering a point with the cutter.}
		\label{fig:cutter-strip-cropping-b}
	\end{subfigure}\hfill
	\caption{(a) A region $A$ that does not have to be covered by the coverage strip of a curved arc $t$, allowing a shortcut of $t$. (b) The intersection of the coverage of a point with the coverage strip of a curved arc implying the situation of \cref{fig:cutter-strip-cropping-a}.}
	\label{fig:cutter-strip-cropping}
\end{figure}
\end{proof}

The following corollary is straightforward.

\begin{corollary}\label{cor:intersection-cutter-strip}
	Let $r$ be a point on the tour $T$, and $t$ be an arc of $T$.
If the coverage of $r$ intersects the convex side $S^+$ of $t$,
then 
$T$ is not optimal, see~\cref{fig:cutter-strip-cropping-b}.
\end{corollary}

Note that in the context of \cref{lemma:straight-line-shorten-tours} shortcutting~$t$ by inserting the segment $N(a)N(b)$ is allowed, i.e., still maintains a covering tour, because the region~$A$ does not have to be covered by $t$. Simultaneously, we obtain another region $A'$ that is now additionally covered by the coverage strip of $N(a)N(b)$, see \cref{fig:cutter-strip-cropping-a}. Intuitively speaking, $A'$ allows for another application of \cref{lemma:straight-line-shorten-tours} to another coverage strip, and so on.

For the remaining details, denote by $p$ an inner point of the polygon, which lies also in the \emph{interior} of the convex side $S^+$, i.e., a point $p$ whose closest point on $t$ is neither $t(0)$ nor~$t(1)$. We show that $p$ needs to have two closest points on $\T$.

\begin{lemma}\label{lemma:two-closest-cutter-points}
	$p$ has at least two closest points on $\T$.
\end{lemma}
\begin{proof}
	As $p$ lies in the interior of a convex side, there is a point $N(p)$ on an arc $t$ and within a distance of $1$ to~$p$. By assumption, $p$ does not lie on the boundary of the polygon, so there is at least a second point $r$ from which the cutter covers~$p$, as otherwise there is an uncovered region, see~\cref{fig:uncovered-portion}. Assume one of these points has a smaller distance to $p$, see~\cref{fig:shorter-distance}. 
	\Cref{cor:intersection-cutter-strip} implies that $\T$ is not optimal, a contradiction.
\end{proof}

\begin{figure}[ht]
	\centering
	\begin{subfigure}{0.49\linewidth}
		\centering
		\includegraphics[page=3, trim={2.5cm 0cm 0cm 0cm},clip, width=.8\linewidth]{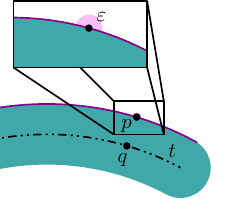}
		\caption{An uncovered portion.}
		\label{fig:uncovered-portion}
	\end{subfigure}\hfill
	\begin{subfigure}{0.49\linewidth}
		\centering
		\includegraphics[page=2, trim={2.5cm 0cm 0cm 0cm},clip,width=.8\linewidth]{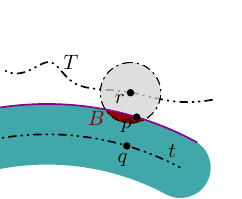}
		\caption{A multi-covered region.}
		\label{fig:shorter-distance}
	\end{subfigure}\hfill
	\caption{(a) An arbitrarily small area (in pink) around $p$ and above $S^+$ that is not covered by the tour. (b) If the point $r$ is closer than a distance of $1$ to $p$, then there is a region (in dark red) that is covered multiple times.}
	\label{fig:two-closest-cutter-points}
\end{figure}

In the context of \cref{lemma:two-closest-cutter-points}, we call $r$ the \emph{successor} of $N(p)$. As $t$ is smooth, $S^+$ is smooth as well. Hence, the tangent to $S^+$ in $p$ is well-defined. We say that a segment \emph{lies orthogonal} to $S^+$ in $p$ when it intersects the tangent to $S^+$ in $p$ orthogonally.

\begin{lemma}\label{lemma:orthogonal}
	The segment between $q$ and its successor~$r$ lies orthogonal to $S^+$ in $p$.
\end{lemma}
\begin{proof}
	Cutter boundaries centered at $q$ and $r$ share at least the point $p$. Assume that $qr$ does not intersect $S^+$ orthogonally in $p$. This implies that at least one of the cutters contains another point from the interior of $S^+$, because $p$ is an inner point of~$S^+$, see \cref{fig:orthogonal-a}. \Cref{cor:intersection-cutter-strip} implies that $\T$ is not optimal, a contradiction.
\end{proof}

\begin{figure}[ht]
	\centering
\begin{subfigure}[t]{0.49\linewidth}
	\centering
	\includegraphics[page=2, trim={2cm 0cm 0cm 0cm},clip,width=.8\linewidth]{figures/coverage-strip-cuddling-orthogonal.pdf}
	\caption{Non-orthogonal intersection.}
	\label{fig:orthogonal-a}
\end{subfigure}\hfill
\begin{subfigure}[t]{0.49\linewidth}
	\centering
	\includegraphics[page=1, trim={2cm 0cm 0cm 0cm},clip, width=.8\linewidth]{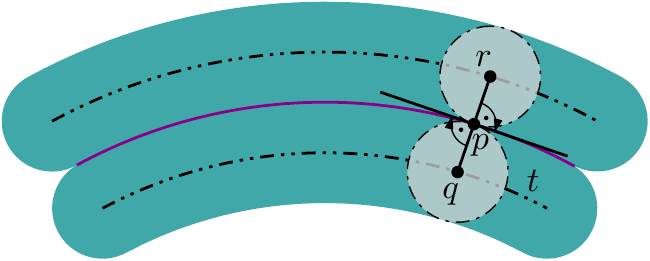}
	\caption{Orthogonal intersection.}
	\label{fig:orthogonal-b}
\end{subfigure}\hfill
	\caption{(a) The situation when $qr$ does not intersect~$S^+$ orthogonally in $p$, i.e., a point from the interior of~$S^+$ lies in the interior of a cutter. (b) The segment~$qr$ intersects $S^+$ orthogonally in $p$.}
	\label{fig:orthogonal}
\end{figure}

Combining \cref{lemma:two-closest-cutter-points,lemma:orthogonal} yield the following.

\begin{corollary}
	$p$ has exactly two closest points on $\T$.
\end{corollary}

Intuitively, the convex side $S^+$ of $t$ is squeezed between two cutters centered in $q$ and $r$. Applying \cref{lemma:straight-line-shorten-tours,,lemma:two-closest-cutter-points,,lemma:orthogonal} yields a proof of \cref{theorem:straight-line-tours-circles}, which we restate here.

\setcounter{@theorem}{0}
\straightlinetours*

\begin{proof}
Applying \cref{lemma:two-closest-cutter-points,,lemma:orthogonal} yields the existence of the successor~$r$~of~$q$, such that $p$, $q$, and $r$ lie in a line that intersects 
$S^+$ of $t$ orthogonally in~$p$.

To ensure that $r$ does not lie on a segment, we slightly perturb $p$ in the interior of $S^+$, such that $r$ lies on an arc. In particular, let $r_1$ and $r_2$ be two successors resulting from two different points $p_1$ and $p_2$ from the interior of $S^+$, such that $r_1$ and $r_2$ lie on segments $s_1$ and $s_2$, respectively. $s_1$ lies orthogonal to $r_1p_1$; otherwise \cref{lemma:straight-line-shorten-tours} implies that the entire tour is not optimal, see \cref{fig:orthogonal-successor-segments}. Analogously, we conclude that $s_2$ lies orthogonal to~$r_2p_2$. Hence, $s_1$ and $s_2$ are different segments, see~\cref{fig:two-different-segments}. As the set of segments is a discrete set, there is a point $p$ in the interior of $S^+$, such that the resulting successor does not lie on a segment, as $S^+$ is continuous.
An analogous argument implies that there is a point $p$ in the interior of $S^+$, such that the successor $r$ is neither a point on a segment nor a start or endpoint of an arc, i.e., $r$ is an inner point of an arc.

\begin{figure}[ht]
	\centering
	\begin{subfigure}{0.49\linewidth}
		\centering
		\includegraphics[page=2,width=.9\linewidth]{figures/coverage-strip-successor-segment.pdf}
		\caption{Non-parallel segments.}
		\label{fig:orthogonal-successor-segments}
	\end{subfigure}\hfill
	\begin{subfigure}{0.49\linewidth}
		\centering
		\includegraphics[page=1,width=.9\linewidth]{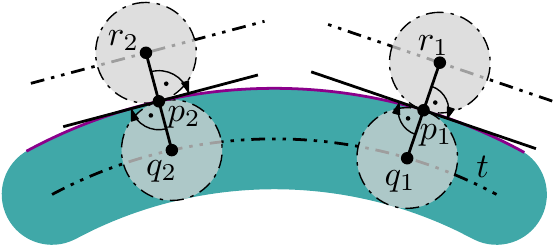}
		\caption{Locally parallel segments.}
		\label{fig:two-different-segments}
	\end{subfigure}\hfill
	\caption{(a) The point $r_1$ lies on a segment that is not parallel to the tangent in $p_1$. (b) The points $r_1$ and $r_2$ lie on two distinct segments that are parallel to the tangents in $p_1$ and $p_2$, respectively.}
	\label{fig:pertubation-of-line}
\end{figure}

This argument of the existence of a successor from the interior of an arc is repeated, resulting in a sequence $R = \{ r, \dots,r'\}$ of successors lying on a line. At some point in this process, a successor $r'$ is reached, fulfilling exactly one of the following cases: (1)~Either the segment $qr'$ crosses the boundary of the polygon (see~\cref{fig:successor-out-of-polygon}), or (2)~the convex side of the arc of~$r'$ touches the polygon boundary in a point~$s$ (see~\cref{fig:successor-hits-polygon}). This means that the line $qr'$ lies orthogonal to the boundary of $\Pol$. In this situation, we denote the intersection point of $qr'$ with the coverage of $r'$ by $s$, which is equal to the intersection between $qr'$ and $\partial \Pol$.

\begin{figure}[ht]
	\centering
	\begin{subfigure}{0.49\linewidth}
		\centering
		\includegraphics[page=1, trim={2.5cm 0cm 0cm 0.33cm},clip,width=.9\linewidth]{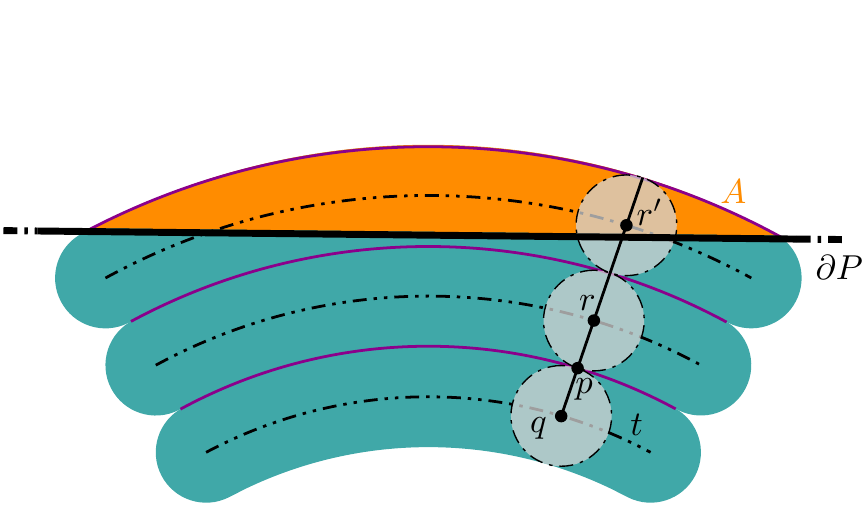}
		\caption{A successor sequence.}
		\label{fig:successor-out-of-polygon}
	\end{subfigure}\hfill
	\begin{subfigure}{0.49\linewidth}
		\centering
		\includegraphics[page=2, trim={2.5cm 0cm 0cm 0.33cm},clip,width=.9\linewidth]{figures/coverage-strip-cuddling-iterate.pdf}
		\caption{Perturbing the sequence.}
		\label{fig:successor-hits-polygon}
	\end{subfigure}\hfill
	\caption{(a) The sequence of constructed successors $R := \{r, \dots,r' \}$ lying in a common ray intersecting the boundary of the polygon $\Pol$ inside the coverage of~$r'$. (b)~Perturbing the successors for avoiding a final successor whose coverage lies tangentially to $\partial \Pol$.}
	\label{fig:successor-in-line}
\end{figure}

In the first case, we obtain a region $A$ that must not be covered by the tour, hence \cref{lemma:straight-line-shorten-tours} implies that $\T$ cannot be optimal.

In the second case, we perturb $p$ including the entire constructed sequence $R$ of successors and the intersection point $s$ of the ray $pr'$ with $\partial \Pol$. This results in the scenario that $s$ is an inner point of $\Pol$, leading to a continuation of the construction of successors until the first case is reached.
This concludes the proof.
\end{proof}

%% file: 03-irrational-tour-vertices.tex
\section{Irrational tour vertices}
\label{sec:irrational-tour-vertices}

In the previous section we saw that optimal lawn mowing tours must be polygonal paths. 
This makes it critical to characterize coordinates of potential tour vertices.
As we show in the following, irrational coordinates may be involved in an optimal tour, indicating possible algebraic (and hence algorithmic) difficulties.



\begin{restatable}{theorem}{irrationalcoordinates}\label{lemma:triangle-irrational-optimal-tour}
For any rational number $0 < \lambda\leq 2$, there are simple polygons $\Pol$ with rational vertex coordinates 
for which an optimal lawn mowing tour contains a vertex with a coordinate of $h := \sqrt{1-{\left(\lambda/2\right)}^2}$.
\end{restatable}

\begin{figure}[ht]
	\centering
	\begin{subfigure}{.49\linewidth}
		\centering
		\includegraphics[page=1, scale=0.75]{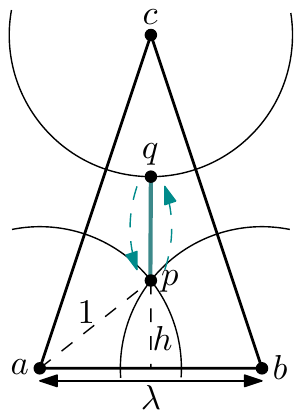}
		\caption{An optimal tour.}
		\label{fig:irrational-example-opt}
	\end{subfigure}\hfill
	\begin{subfigure}{.49\linewidth}
		\centering
		\includegraphics[page=2, scale=0.75]{figures/irrational_solutions_general.pdf}
		\caption{A suboptimal tour.}
		\label{fig:irrational-example-not-opt}
	\end{subfigure}\hfill
  \caption{(a) A polygon $\Pol = (a,b,c)$ with the optimal lawn mowing tour, connecting $p$ and $q$. (b)~A possible lawn mowing tour without using an irrational point.}
	\label{fig:irrational-example}
\end{figure}

\begin{proof}
	For $0 < \lambda \leq 2$, let $\Pol = (a,b,c)$ be the polygon with $a=(0,0)$, $b=(\lambda,0)$, and $c=\left(\frac{\lambda}{2}, 2\right)$, see~\cref{fig:irrational-example}. It is easy to see that any lawn mowing tour must have at least one point that touches the unit circles $C_a, C_b$, and $C_c$ that are centered at the vertices of $\Pol$.  
	
	We show that the point $p$ on the intersection of $C_a$ and $C_b$ within $\Pol$ has to be part of the optimal lawn mowing tour, see \cref{fig:irrational-example-opt}. 
	
	Assume that $\T$ is an optimal tour, and let $q$ be some point on $\T$ with $q \in \partial C_c$. For the sake of contradiction, assume that $p \not \in \T$.
	Thus, there exist two points $p_\ell, p_r \in \T$, with $p_\ell \in \partial C_a$, and $p_r \in \partial C_b$, see~\cref{fig:irrational-example-not-opt}.
	As $\T$ connects $p_\ell, p_r$, and $q$, the length of the tour is
	$\ell(\T) \geq d(p_\ell, p_r) + d(p_\ell, q) + d(p_r, q)$.
	Independent of the choice of the points $p_\ell, p_r$, and~$q$, we observe that
	$d(p_\ell, q) + d(p_r, q) > 2 d(p, q)$.
	Moreover, $\T'=pqp$ is a feasible tour with length $\ell(\T') < \ell(\T)$, contradicting the assumption that $\T$ is an optimal tour.
	
	It remains to show that $p$ has an irrational coordinate, and that $\T'$ is the unique optimal tour.
  The three points $a,b$, and $p$ describe an isosceles triangle of base length $\lambda$, thus the $y$-coordinate of $p$ is $\sqrt{1-{\left(\lambda/2\right)}^2}$. The point $q$ that is closest to $p$ clearly lies on the vertical line $x= \frac{\lambda}{2}$. It is easy to see that the tour that connects the points $p = \left(\frac{\lambda}{2}, \sqrt{1-{\left(\lambda/2\right)}^2}\right)$ and $q = \left(\frac{\lambda}{2}, 1\right)$ covers $\Pol$ and has length $\ell(\T') = 2\left(1-\sqrt{1-{\left(\lambda/2\right)}^2}\right)$.
	Because $p$ has to be a vertex of the tour, and any other choice of $q$ results in a longer tour, $\T'$ is the unique optimal lawn mowing tour for $\Pol$.
\end{proof}
  
As a consequence, bounded accuracy in computing coordinates may become an issue, motivating the concept of
$\varepsilon$-\emph{robust} coverage, for which any point must be covered by a portion of the cutter $C$ that is
at least $\varepsilon$-interior to $C$, corresponding to a CETSP solution with neighborhood size $r-\varepsilon$.
Furthermore, the algebraic difficulty motivates the following.

\begin{problem}
\label{prob1}
Is the LMP $\exists\R$-complete?
\end{problem}

%% file: 04-cetsp-approach.tex
\section{A primal-dual approach to Lawn Mowing}
\label{sec:primaldual}

Now we describe our primal-dual approach to computing good tours for the
LMP, based on (i)~iteratively selecting
finite subsets $W$ of points from the uncovered portion of $P$,
(ii)~covering $W$ by an optimal solution $T_W$ for the
Close-Enough TSP for $W$ to obtain a lower bound, (iii)~enhancing $T_W$ to a feasible covering tour of $P$ to obtain an upper bound.

In the following, we describe basic mechanisms and analytical implications;
further technical details that affect the practical performance
are described in the following \cref{sec:exp}.


\subsection{Witness sets and Close-Enough TSP}

At each stage of our algorithm, we use a finite set $W$ of \emph{witnesses},
in $P$. It is straightforward to see that these induce lower bounds on the length of
an optimal lawn mowing tour.

\begin{theorem}
\label{th:LB}
For $W_1\subset W_2\subset P$, let $T_{W_1}$, $T_{W_2}$, and $T_P$ be optimal lawn mowing tours
for $W_1$, $W_2$, and $P$, respectively. Then $\ell(T_{W_1})\leq\ell(T_{W_2})\leq\ell(T_P)$.
\end{theorem}

Because all points $w\in W$ are covered by a cutter $C$ of radius $r$ if and only if
the center of $C$ visits all disks of radius $r$ centered at $w$, such a
\emph{Close-Enough TSP} tour for $W$ corresponds to a lawn mowing tour
for $W$, and thus a lower bound for the LMP on $P$.
Therefore, a sequence of enhanced witness sets leads to an increasing sequence of
lower bounds for $\ell(T_P)$; if some $T(W_i)$ is a feasible LMP tour for $P$, it
is optimal.

%
%

Arkin et al.~\cite{arkin1993lawnmower, Arkin2000} noted that any LMP tour may need a pseudopolynomial
(but not polynomially bounded) number of tour vertices even for a
square-shaped lawn with diameter $d$ and $d>>r$. 
For $\varepsilon$-\emph{robust} coverage (as defined in \cref{sec:irrational-tour-vertices}),
we can establish a finite upper bound on the number of necessary witnesses that is independent of $n$, and
holds even for general sets $P\subset \R^2$, not necessarily just polygons or connected sets.

%
\begin{theorem}\label{theorem:witness-size-less-r}
Consider a region $\Pol \subset \R^2$ of diameter $d$, a circular cutter $C$ of radius $r$ and some $0<\varepsilon\leq \frac{r}{2}$.
Then there is a witness set $W$ of size bounded by $|W| \leq \ceil{\frac{2d^2}{\varepsilon^2}} \in O(d^2)$,
for which any $\varepsilon$-\emph{robust} CETSP tour is also feasible for the LMP\@.
\end{theorem}

\begin{proof}
A bounding box $B$ of $P$ has dimensions at most $d\times d$, so we can subdivide $B$ by an orthogonal mesh of
width $\varepsilon/\sqrt{2}$, consisting of not more than $\ceil{\frac{2d^2}{\varepsilon^2}}$ grid cells, $C_{ij}$,
each of diameter at most $\varepsilon$. Whenever $P\cap C_{ij}\neq\emptyset$, picking an arbitrary witness
point from $P\cap C_{ij}$ ensures that each such cell (and hence, all of $P$) is covered by
any $\epsilon$-robust CESTP tour.
\end{proof}


For general (and not just robust) witness sets, it is easy to see that
all extreme points of $P$ must be contained; if $P$ is a convex polygon,
this includes all $n$ vertices.  Thus, no such upper bound exists in
terms of only $d$ and $r$. For other types of regions $P$ (such as disks), no finite
bound of any kind exists, as the whole perimeter needs to be contained in a witness set.

\begin{problem}
\label{prob2}
For a polygon $P$ with $n$ vertices and diameter $d$, and a circular cutter of radius $r$,
is there a witness set $W$ of size polynomial in $n$, $d$, $1/r$, such that
an optimal CETSP solution for $W$ is a feasible lawn mowing tour of $P$?
\end{problem}


\subsection{Lower bounds: convergence}
As we enhance the witness set, we can ensure that the maximum distance of uncovered points from
an ensuing CETSP tour converges to zero. This can be achieved by a relatively
dense set of witness points (as in \cref{theorem:witness-size-less-r}) or
(more economically) by iteratively reducing this maximum distance when
necessary. We get the following straightforward implication.

\begin{lemma}
\label{le:LB}
Let $W\subset P$ be a finite set of witness points, such that no point of $P$ has a distance
larger than $\delta$ from a witness point, and let $T_W$ be a feasible CETSP tour of $W$. Then
  \[d_{\max}:=\max_{p\in P} \min_{q\in T_W\oplus C} d(p,q)\leq \delta.\]
\end{lemma}

In the remainder of this subsection, we will repeatedly apply the following theorem.

\begin{theorem}[Fekete and Pulleyblank~\cite{fp-tbms-98}]
\label{FP}
Let $E$ be a connected planar arrangement of edges and
$C$ be a disk of radius $r$.
If $L$ is the total length of the edges in $E$, there is a closed roundtrip of
length at most $2L + 2\pi r$ that visits all points of the boundary of $E \oplus C$.
\end{theorem}

In the sequel, $E$ will be the edges of a tour, leading to the following.

\begin{theorem}
\label{th:3}
Let $W$ be a witness set, and let $T_W$ be an optimal CETSP tour of $W$. Let the
distance between uncovered points and covered region be
bounded by the cutter radius $r$, i.e.,
$d_{\max}:=\max_{p\in P} \min_{q\in T_W\oplus C} d(p,q)\leq r$.
Then $P$ has a lawn mowing tour of length at most $3\ell(T_W)+2\pi r + 2r$.
\end{theorem}

\begin{proof}
By assumption, any uncovered point $p\in P$
is within distance $r$ from the boundary of
$E \oplus C$. Therefore, all points that are not covered by $T_W$ will be covered
by a tour of the boundary of $E \oplus C$; by \cref{FP}, there is such
a tour of length at most $2\ell(T_W) + 2\pi r$. Concatenating this tour with $T_W$
(at additional cost at most $2r$) achieves full coverage of $P$ at cost
at most $3\ell(T_W)+2\pi r + 2r$, as claimed.
\end{proof}

Furthermore, we can conclude the following.

\begin{lemma}
Let $W_0\subset W_1\subset\ldots\subset W_i\subset\ldots\subset P$ be a
sequence of witness sets, with $T_{W_i}$ being the corresponding sequence of optimal CETSP tours of $W_i$.
If the maximum distance $d_{\max}^{(i)}:=\max_{p\in P} \min_{q\in T_{W_i}\oplus C} d(p,q)$
of uncovered points to the covered region converges
to zero, then the area of $P\setminus T_{W_i}\oplus C$ converges to zero.
\end{lemma}

\begin{proof}
By assumption, the sequence $d_{\max}^{(i)}$ converges to zero, so there
is an index $k$ such that $d_{\max}^{(i)}\leq r$ for all $i\geq k$. Then by
\cref{th:LB}, all elements of the sequence
$\ell(T_{W_i})$ are lower bounds of any feasible LMP tour,
so by \cref{th:3}, they satisfy $\ell(T_{W_i})\leq 3\ell(T_{W_k})+2\pi r + 2r$,
i.e., remain bounded from above.
Furthermore, all points in the set $P\setminus T_{W_i}\oplus C$
of uncovered points by $T_{W_i}$ are within distance $d_{\max}^{(i)}$
of the boundary of $T_{W_i}\oplus C$; by \cref{FP},
this boundary has length of at most $L=2\ell(T_{W_i})+2\pi r$. Moreover,
the area of all points within a distance $d_{\max}^{(i)}$
of any curve of length $L$ is bounded by $A_i:=(2L+2\pi)d_{\max}^{(i)}$,
so the total uncovered area is bounded by $A_i\leq (4\ell(T_{W_i})+ 4\pi r + 2\pi)d_{\max}^{(i)}\leq
(12 \ell(T_{W_k}) + 12\pi r + 8r + 2\pi)d_{\max}^{(i)}$.
Because $(12 \ell(T_{W_k}) + 12\pi r + 8r + 2\pi)$ is constant and $d_{\max}^{(i)}$
by assumption converges to zero, we conclude that $A_i$ also converges to zero,
implying the claim.
\end{proof}

We summarize.

\begin{theorem}
By picking an appropriate sequence of witness sets $W_i$, we can guarantee that
for the ensuing sequence of optimal CETSP solutions $T(W_i)$, the following holds.
\begin{itemize}
\item The maximum distance of uncovered points to the covered region converges to zero.
\item The total uncovered area converges to zero.
\end{itemize}
\end{theorem}

\subsection{Solving CETSP instances}
Solving CETSP instances is more challenging than solving TSP instances.
In 2016, Coutinho et al.~\cite{coutinho2016branch} proposed an
algorithm based on branch-and-bound and \emph{Second-Order Cone Programming}.
Each branch-and-bound node is associated with a partial
tour that visits the given subset of vertices in particular order.  At the root
node, the algorithm chooses three vertices to generate an initial sequence.
The problem of computing the exact coordinates from a predefined sequence can
be formulated as a {Second-Order Cone Problem} (SOCP).  If the solution at the
root node is feasible, the solution is optimal and the algorithm terminates.
Otherwise, the algorithm branches into three subproblems, one for every edge in the current solution.
At each edge a currently uncovered vertex is inserted.
A node is pruned if the cost is at least equal to the best known upper bound or if its associated solution is feasible.

\begin{minipage}{\columnwidth}
  \begin{alignat}{2}
      \text{minimize} \quad && \sum_{k=1}^q z_k
  \end{alignat}
  subject to
  \vspace{-.15cm}
  \begin{alignat}{2}
       x_{i_{k-1}} - x_{i_k} = w_k &\qquad & k = 0,\dots, q \label{socp:cetsp-aux-1} \\
       y_{i_{k-1}} - y_{i_k} = u_k  && k = 0,\dots, q \\
       \bar{x_k} - x_k = s_k  && k = 0,\dots, q \\
       \bar{y_k} - y_k = t_k && k = 0,\dots, q \label{socp:cetsp-aux-4} \\
       w_k^2+u_k^2 \leq z_k^2  && k = 0,\dots, q \label{socp:cetsp-bound-length}\\
       s_k^2+t_k^2 \leq r_k^2  && k = 0,\dots, q \label{socp:cetsp-circle-constr}\\
       z_k \geq 0; \hspace{1mm}x_k,y_k,w_k,u_k,s_k,t_k \hspace{1mm}\text{free} && k = 0,\dots, q \nonumber
  \end{alignat}
\end{minipage}
\vspace{5mm}

The formulation by Coutinho et al.~\cite{coutinho2016branch} is based on work of Mennell~\cite{mennell2009heuristics}.
Let $S=\{i_0,\dots,i_q\},\,q<n$ be a partial sequence of vertices during the execution of the branch-and-bound algorithm.
The desired solution must provide exact values for the coordinates $(x_{i_k}, y_{i_k}),\,k=0,\dots,q$ such that the length of the partial tour is minimized. Denoting $i_{-1}=i_q$, the formulation is as follows.

The SOCP uses variables $z_k$ to represent the distance between subsequent
vertices $i_{k-1}$ and $i_k$ in the given sequence $S$.  To compute
the distance, four auxiliary variables $w_k,u_k, s_t$, and $t_k$ are introduced
in Constraints~(\ref{socp:cetsp-aux-1})--(\ref{socp:cetsp-aux-4}) that represent
differences of coordinates that are used to calculate Euclidean distances.
Constraints~(\ref{socp:cetsp-bound-length}) define the length
of the edge connecting subsequent vertices.
Constraints~(\ref{socp:cetsp-circle-constr}) ensure that the hitting points will
lie within their respective circles around $v_i$.  Andersen et
al.~\cite{andersen2003implementing} showed that a SOCP can be solved in
polynomial time.  The branch-and-bound approach of Coutinho et
al.~\cite{coutinho2016branch} makes use of well-known optimization software that
is also capable of addressing this class of problems.

\subsection{Upper bounds: obtaining feasible solutions}
\label{subsec:upper}
Given is a polygon $\Pol$, a circular cutter $C$ and a maximum number of iterations. We start by initializing the witness set $W_0$ with some points $p \in \Pol$.


In each iteration $i$, we start with generating a CETSP \emph{start tour} $\T_i$ from $W_i$;
see \cref{fig:eval:witnessexamples:ubexamples} for families of examples.
If $T_i$ is a feasible LMP tour, it is optimal; otherwise,
we can extract the \emph{uncovered regions} from $\Pol \setminus (\T_i \oplus C)$.

As shown in \cref{fig:eval:witnessexamples:ubexamples}
for each of the lower bounds,
we then execute the following procedure to modify $T_i$ until the tour becomes feasible.
For every uncovered region $R_j$, we construct a witness set $W_{R_j}$
and fix the closest point $p_j \in \partial\T_i$ to $R_j$
to be visited by a tour $T_i$ computed by the CETSP solver.

After the tour $T_{i}$ is computed, we add $W_{R_j}$ to obtain $W_{i+1}$,
so that the next start tour covers the points in $W_{R_j}$ from the beginning.
(See \cref{sec:generating-witness-sets} and again~\cref{fig:eval:witnessexamples:ubexamples}
for practical strategies for choosing $W_{R_j}$.)
Subsequently, we extend $\T_i$ with $\T_{R_j}$, by connecting both tours via the point
$p_j$.
Once we obtain a feasible tour $\T_i$ (such as in the examples of \cref{fig:eval:witnessexamples:ubexamples}),
we can update the best known solution so far and continue with the next iteration for lower
and upper bounds, see \cref{algo:cetsp-feasible}.

%% file: 05-experiments.tex
\begin{figure}[h]
  \centering
  \begin{subfigure}[t]{.5\columnwidth}
  \centering
    \includegraphics[width=\linewidth]{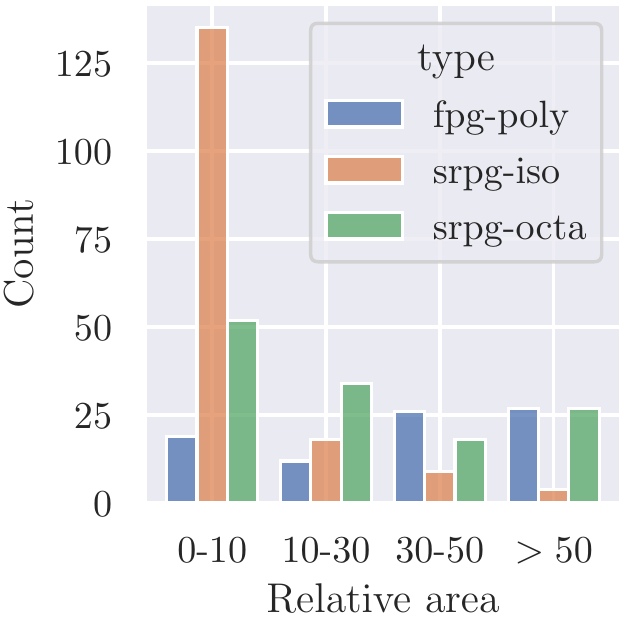}
    \caption{Size distribution.}
  \end{subfigure}\hfill
   \begin{subfigure}[t]{.5\columnwidth}
  \centering
    \includegraphics[width=\linewidth]{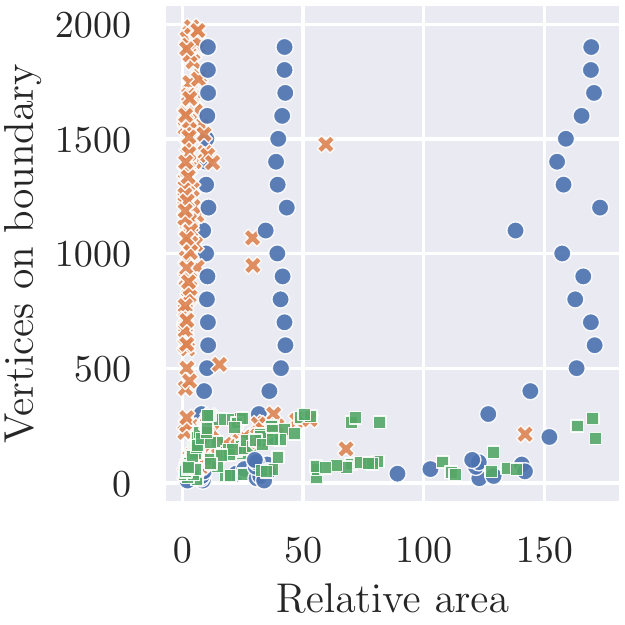}
     \caption{Polygon complexity.}
  \end{subfigure}
  \caption{Distribution of instances, subdivided into different types. See \cref{fig:distribution-enlarged} for an enlarged version of the small instances in (b).}
\label{fig:distribution}
\end{figure}

\begin{figure*}[t]
  \centering
    \includegraphics[width=.32\linewidth]{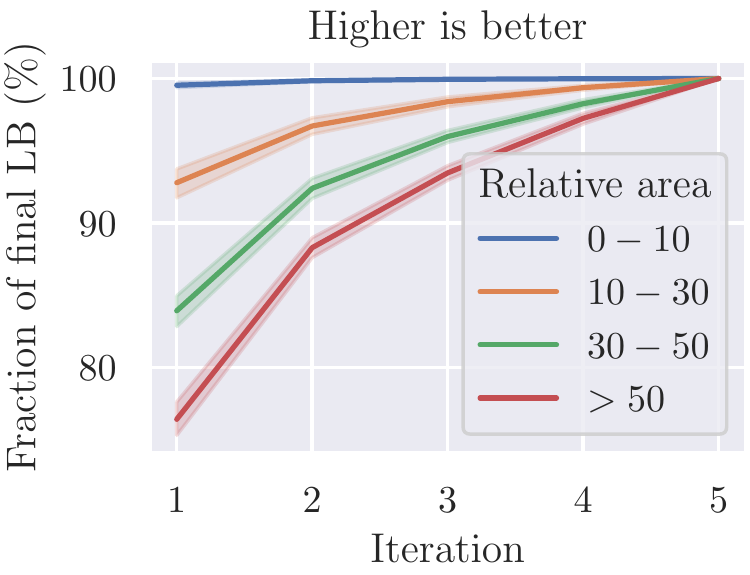}
    \includegraphics[width=.32\linewidth]{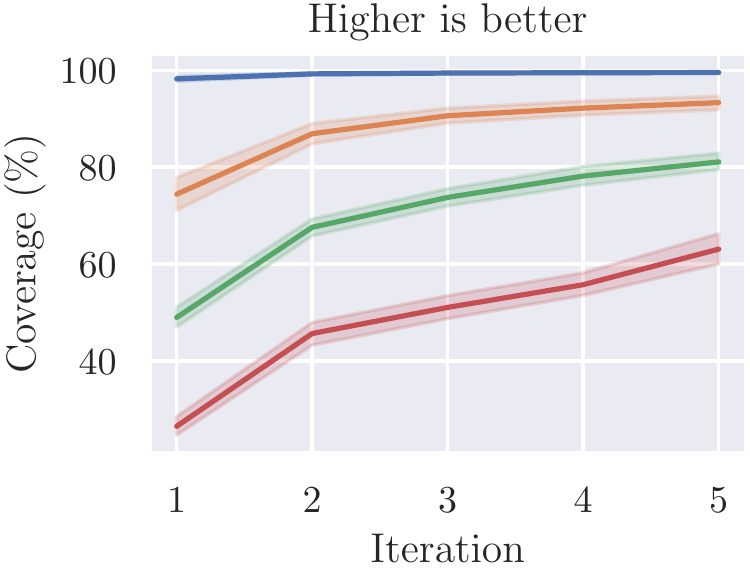}
    \includegraphics[width=.32\linewidth]{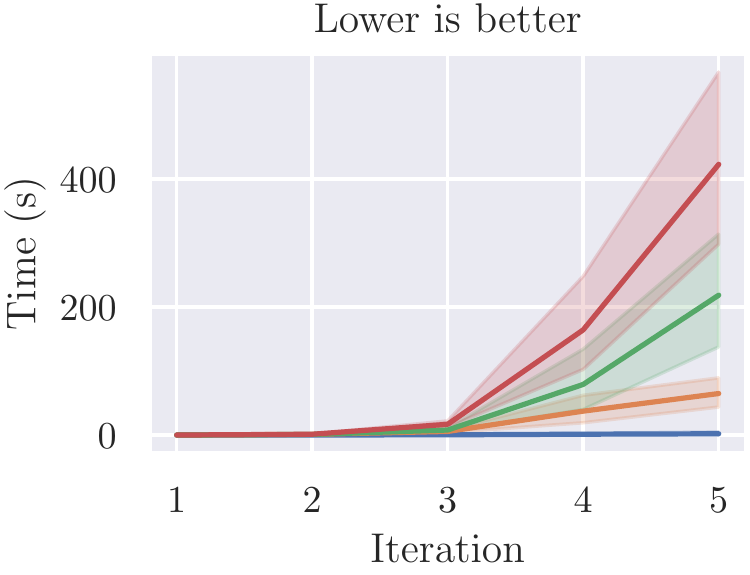}
  \caption{%
    Improvement of lower bound and coverage over iterations, where each iteration extends the witness set.
    The lower bound is measured in relation to the final lower bound. All three plots share the same legend.
    For small instances, the first lower bound is already very close to the final upper bound, while for medium instances, it is still around \SI{93}{\percent} of the final lower bound.
    The larger the instances, the stronger the increase with each iteration. We get a similar result for the percentage of covered area.
    }
  \label{fig:eval:lb:iterations}
\end{figure*}

\begin{figure*}[t]
  \centering
    \includegraphics[width=0.49\linewidth]{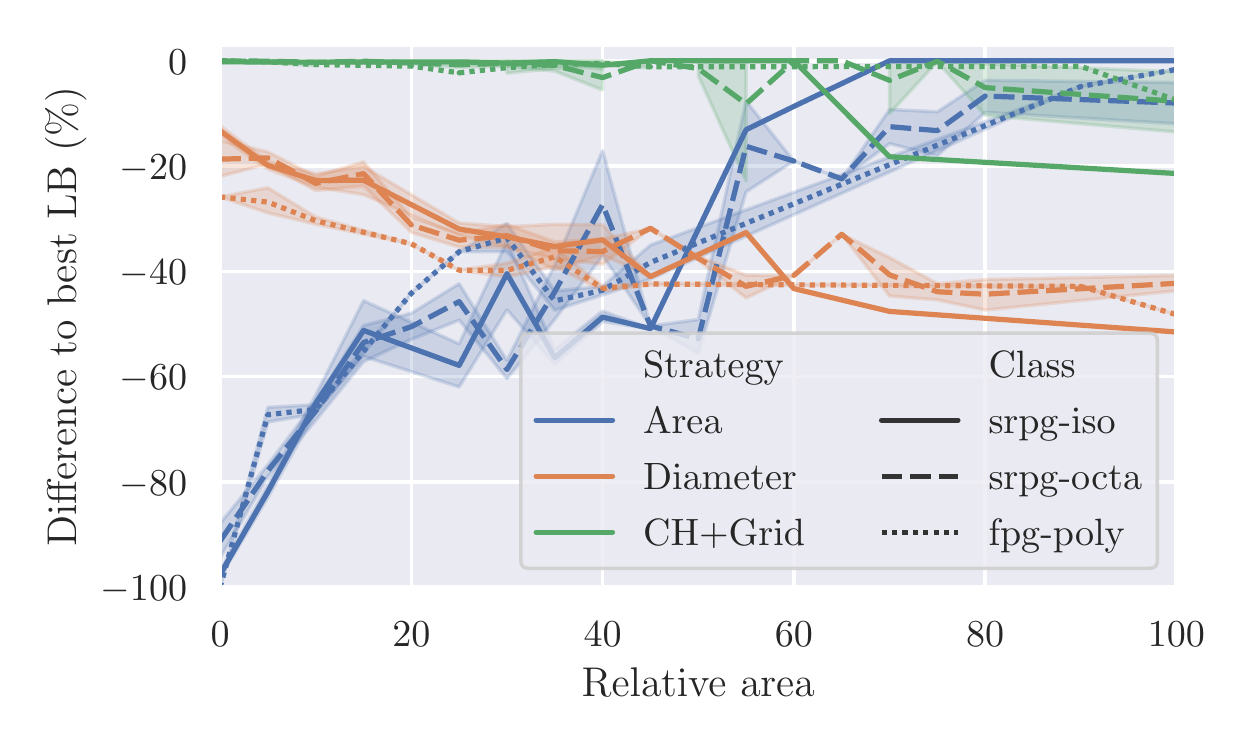}
    \includegraphics[width=0.49\linewidth]{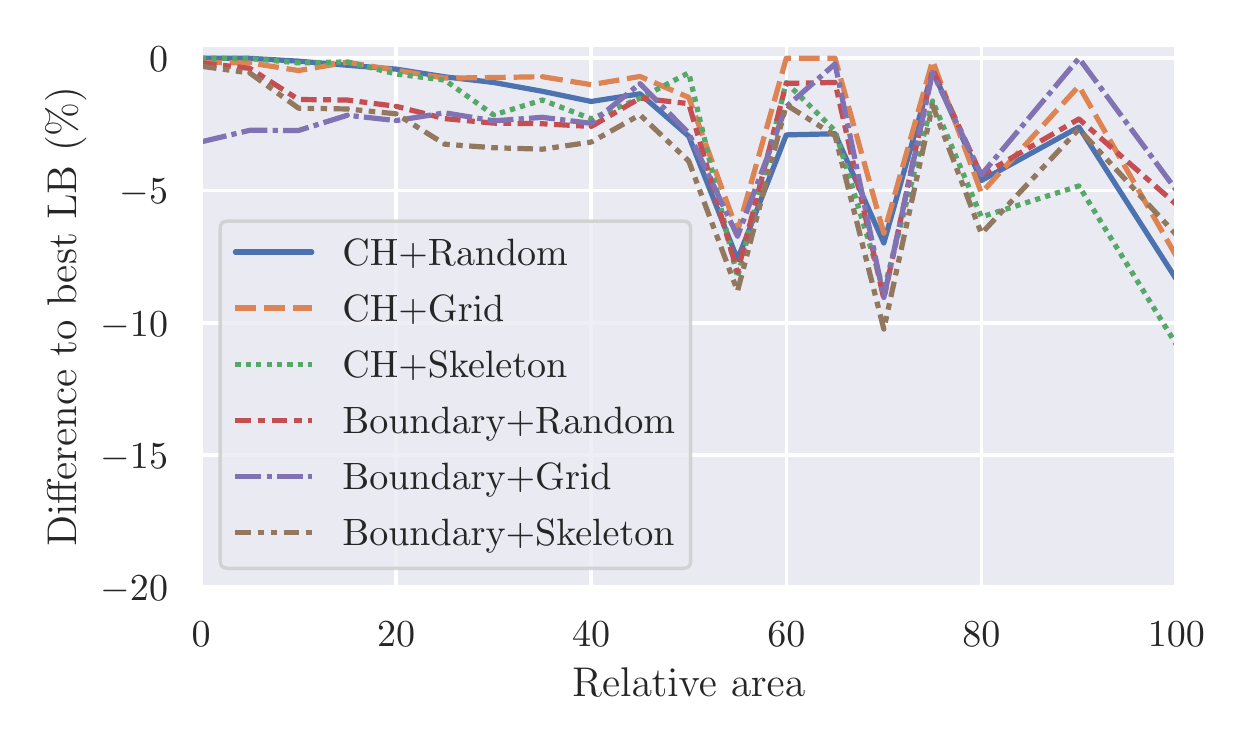}
    \vspace{-0.25cm}
  \caption{%
    Comparison of lower bounds based on area or diameter with the lower bounds provided by CETSP for different witness strategies (left).
    Up to a relative area of \num{50}, the simple lower bounds are more than \SI{20}{\percent} smaller than the CETSP on the convex hull and random witnesses for all classes of instances.
    Between the different witness strategies (right), the deviation is relatively small.
    Just the strategies that start with a subset of the convex hull have a small advantage.}
    \label{fig:eval:lb}
\end{figure*}

\section{Experiments and evaluations}
\label{sec:exp}

Now we describe our practical implementation and various
algorithm engineering aspects, along with our experimental study for demonstrating the practical
usefulness of our theoretical concept.
All experiments were carried out on a regular desktop workstation with an
AMD Ryzen 7 5800X ($8\times\SI{3.8}{\GHz}$) CPU and \SI{128}{\giga\byte} of RAM\@.
The code and data are publicly available\footnote{Source code and data: \url{https://github.com/tubs-alg/near-optimal-lawn-mowing-tours}}.
Consider \cref{fig:eval:witnessexamples:ubexamples} for practical illustrations of
the progression of lower bounds for different witness strategies,
and for the corresponding upper bounds.
For carrying out our experiments and evaluations, we used the \emph{FPG}, \emph{SRPG ISO} and \emph{SRPF OCTA} instances from the Salzburg
Database of Geometric Inputs~\cite{eder2020salzburg}. See \cref{fig:distribution} for the overall
distribution, and \cref{fig:benchmarks} for a sample of individual instances.
We considered polygons with up to $n=2000$ vertices, in combination with
cutters of varying size.
Overall, this resulted in several hundred instances.
As it turned out, the parameter \emph{relative area}, i.e., the ratio between the areas of the convex hull
of $P$ and the cutter $C$, was more significant for
the difficulty of an instance than the vertex number~$n$, which is why we mostly focused
on the variance in relative area.

\subsection{Generating witness sets}\label{sec:generating-witness-sets}

As we established in the previous section, enhancing the witness set leads to
convergence of the resulting lower bound tours, at least in terms of covered area
and distance to uncovered points. In a practical setting, solving the involved
CETSP instances becomes a bottleneck, so it becomes crucial to achieve
good lower bounds with witness sets of limited size.

For the initial witness set $W_0$, we focus on the exterior 
of the polygon,
as minimizing the length in further iterations will automatically pull the CETSP tour towards the interior;
so we choose up to \num{15} witnesses either from the convex hull, or the boundary of $\Pol$.
If there are more than \num{15} candidate points, we apply a greedy max-min-dispersion to achieve a well-distributed witness set.
This may already lead to optimal solutions for small instances.

For further iterations, we insert \num{10} more witnesses into regions 
that are not yet covered.
For multiple uncovered regions, we do a random assignment, according to region size.
Within each uncovered region $R_j$, we select the assigned number of new witnesses either (i) randomly, (ii) from a grid, or (iii) from the straight skeleton~\cite{aichholzer1996novel}.
To avoid strongly clustered sets (which are bound to occur for the straight skeleton),
we employ a dispersion technique.
To prevent new witnesses from being close to the boundary but deep inside of the uncovered regions,
we cluster the candidate points using $k$-means, and select a random witness from each cluster.

The two strategies for choosing the initial witness set and the three strategies for extending it result in six general strategies, some of which are shown in the first three columns of \cref{fig:eval:witnessexamples:ubexamples}.

\begin{figure*}[t]
  \centering\hfill
  \begin{subfigure}[b]{.66\linewidth}
    \centering
    \includegraphics[width=\linewidth]{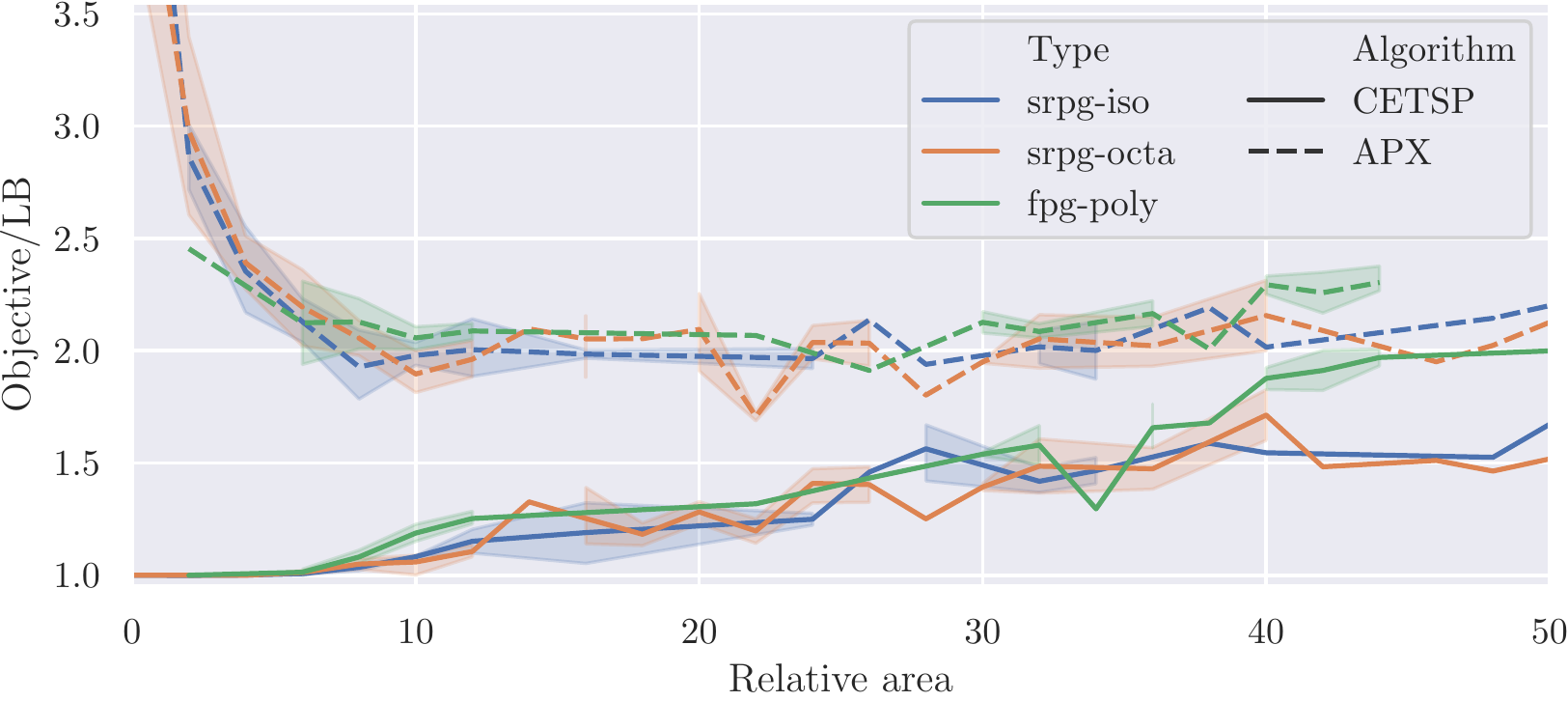}\\
    \caption{}\label{fig:eval:ub}
  \end{subfigure}\hfill
  \begin{subfigure}[b]{.15\linewidth}
    \centering
    \includegraphics[width=\linewidth]{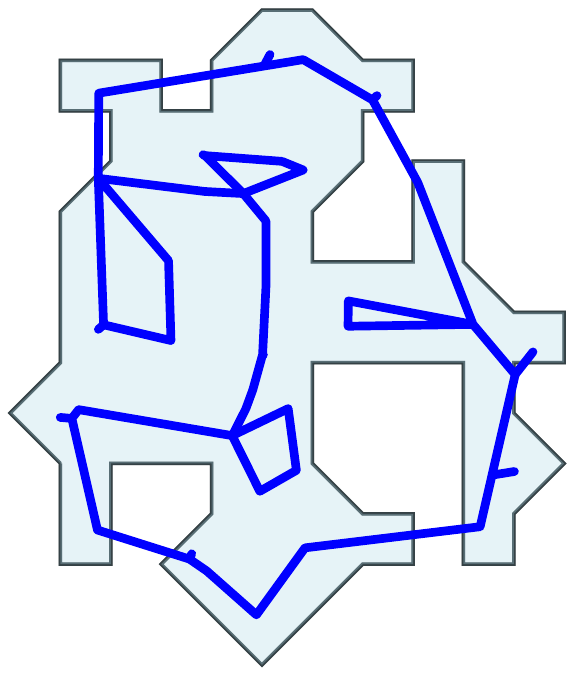}
    \vspace{1cm}
    \caption{}
    \label{fig:approx-comparison:cetsp}
  \end{subfigure}\hfill
  \begin{subfigure}[b]{.15\linewidth}
    \centering
    \includegraphics[width=\linewidth]{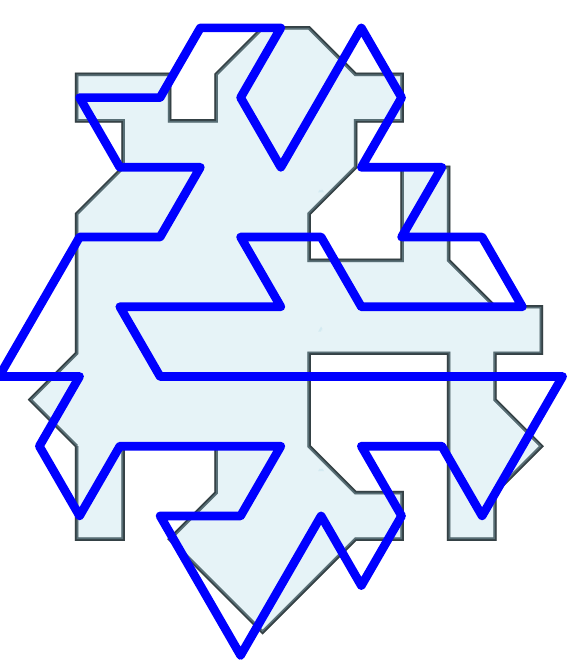}
    \vspace{1cm}
    \caption{}
    \label{fig:approx-comparison:approx}
  \end{subfigure}\hfill
  \vspace{-0.25cm}
  \caption{%
    \textbf{(a)} Comparison of the solution quality of CETSP approach (plus heuristic) and approximation algorithm (APX).
    We compare the best solution of any of the witness strategies.
    Our approach often yields optimal solutions for instances up to a relative area of \num{6}, see \cref{fig:eval:optexamples} for examples.
    \textbf{(b)} An upper bound from our approach with a value of \num{65.35}.
    \textbf{(c)} An upper bound from the approximation of Arkin et al.~\cite{Arkin2000} with a value of \num{88.33}.
    }
    \vspace{-.7cm}
\end{figure*}
\begin{figure*}
  \centering
  \includegraphics[width=.8322\linewidth]{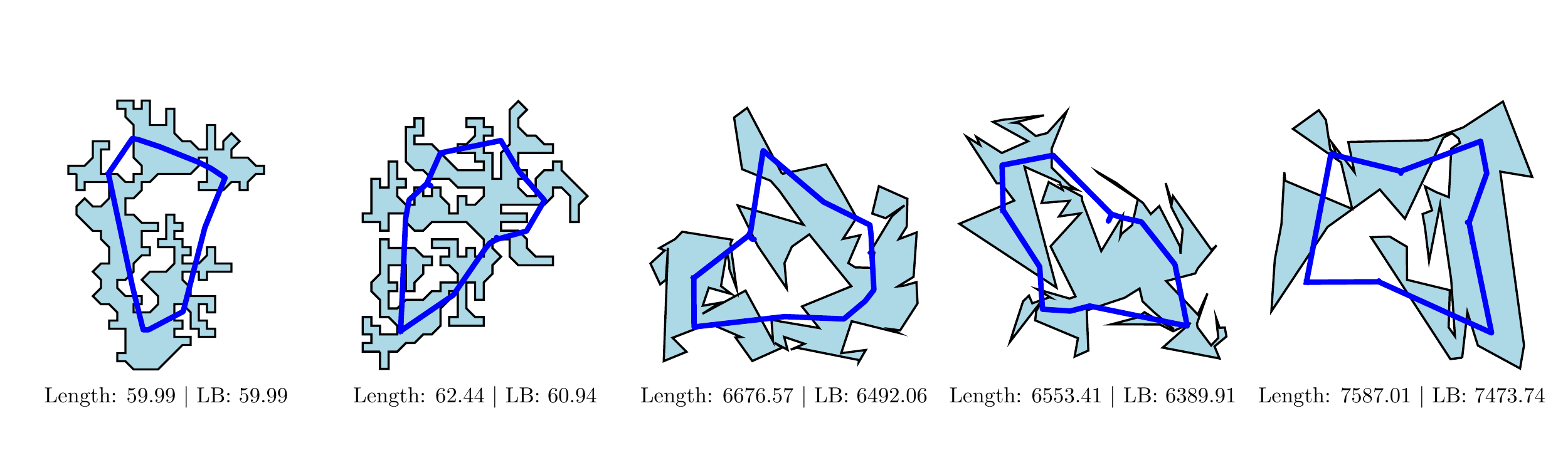}
  \vspace{-0.25cm}
  \caption{Examples of instances solved to (near) optimality ($\leq \SI{3}{\percent}$ above the lower bound).}
  \label{fig:eval:optexamples}
\end{figure*}

\subsection{Practical quality of lower bounds}

In each iteration, our algorithm computes an optimal CETSP tour on the current
witness set; this results in an increasing sequence of lower bounds (due to \cref{th:LB}),
which are asymptotically tight, at least in terms of covered area
and distance to uncovered points.
Thus, computing exact CETSP solutions becomes critical. While we could still solve clustered instances
with up to \num{1000} points to optimality, they do not necessarily yield good
lower bounds for the LMP\@; on the other hand, well-distributed witness sets much beyond 50 points were harder to
solve optimally, but still yielded better bounds, illustrating
the delicate balance between the structure of witness sets, CETSP solvability,
and quality of lower bounds.

In the following, we use the best lower bound obtained within \num{5} iterations.
The calls to the CETSP solver are limited to \SI{30}{\minute}, after which the
best lower bound 
is returned. \Cref{fig:eval:lb:iterations} shows the progress of
lower bound and coverage. 
With each iteration, the CETSP gets harder to solve, and the
improvements on lower bound and coverage decrease. As shown in \cref{fig:eval:lb} (left),
the resulting lower bound after the final iteration is notably better than the lower bound derived
from the diameter or area of the polygon.
This is true for all three classes of instances, up to a relative area of about \num{50},
after which the CETSP strategy gets bogged down by too many witnesses.
The lower bounds of the individual witness strategies are surprisingly close, as shown in \cref{fig:eval:lb} (right).
Starting with the convex hull gives a small advantage; otherwise, the best strategies vary strongly for the instances.
Example runs of some strategies over 5 iterations are shown in \cref{fig:eval:witnessexamples:ubexamples}.
The best way to improve these (and other) lower bounds would be to add further iterations, at the expense of runtime.

\subsection{Practical quality of upper bounds}
For small instances, the CETSP procedure frequently returns feasible (and hence, optimal) LMP tours.
Otherwise, we perform up to \num{5} iterations in which we add \num{40} witnesses into each uncovered
area and connect the tour to the previous tour, as described in \cref{subsec:upper}.
For medium-sized instances, this strategy often leads to feasible solutions that
are significantly better than the approximation algorithm;
see \cref{fig:eval:ub,fig:approx-comparison:cetsp,fig:approx-comparison:approx} for an example.
For instances up to a relative area of around \num{6}, we often obtain optimal solutions, or solutions with a negligibly small gap.
For instances up to a relative area of around \num{30}, we are at most \SI{50}{\percent} above the lower bound,
and only slightly worse for \num{40}.
For larger instance sizes, \cref{th:3} still bounds the relative gap to \SI{200}{\percent}, but the CETSP
computation becomes a bottleneck.
For instances up to a relative area of \num{40}, the approximation algorithm
turns out to be relatively stable at \SI{100}{\percent} above the lower bound,
which is considerable better than the
classical worst-case approximation guarantee by Arkin et al.~\cite{Arkin2000},
which is at $2\sqrt{3}+\varepsilon\approx3.46$ even when disregarding
some finer technical aspects:
That method is based on combining
a proximity approximation for considering region boundary and interior
(incurring an \emph{asmyptotic}  factor of 3), a grid conversion
(incurring a factor of $2/\sqrt{3}$), and a TSP approximation
(incurring a theoretical factor of $1+\varepsilon$ when using
a geometric PTAS, and a larger factor when using a more practical
approximation method); this results in an overall factor no better than
$2\sqrt{3}+\varepsilon\approx3.46$, and even more when accounting
for a non-asymptotic approximation factor with an additive term (see~\cite{Arkin2000})
and PTAS limitations.
For our experiments, however, the TSP instances in the approximation algorithm
were solved to optimality.

\begin{figure*}
  \centering
  \includegraphics[width=0.8322\linewidth]{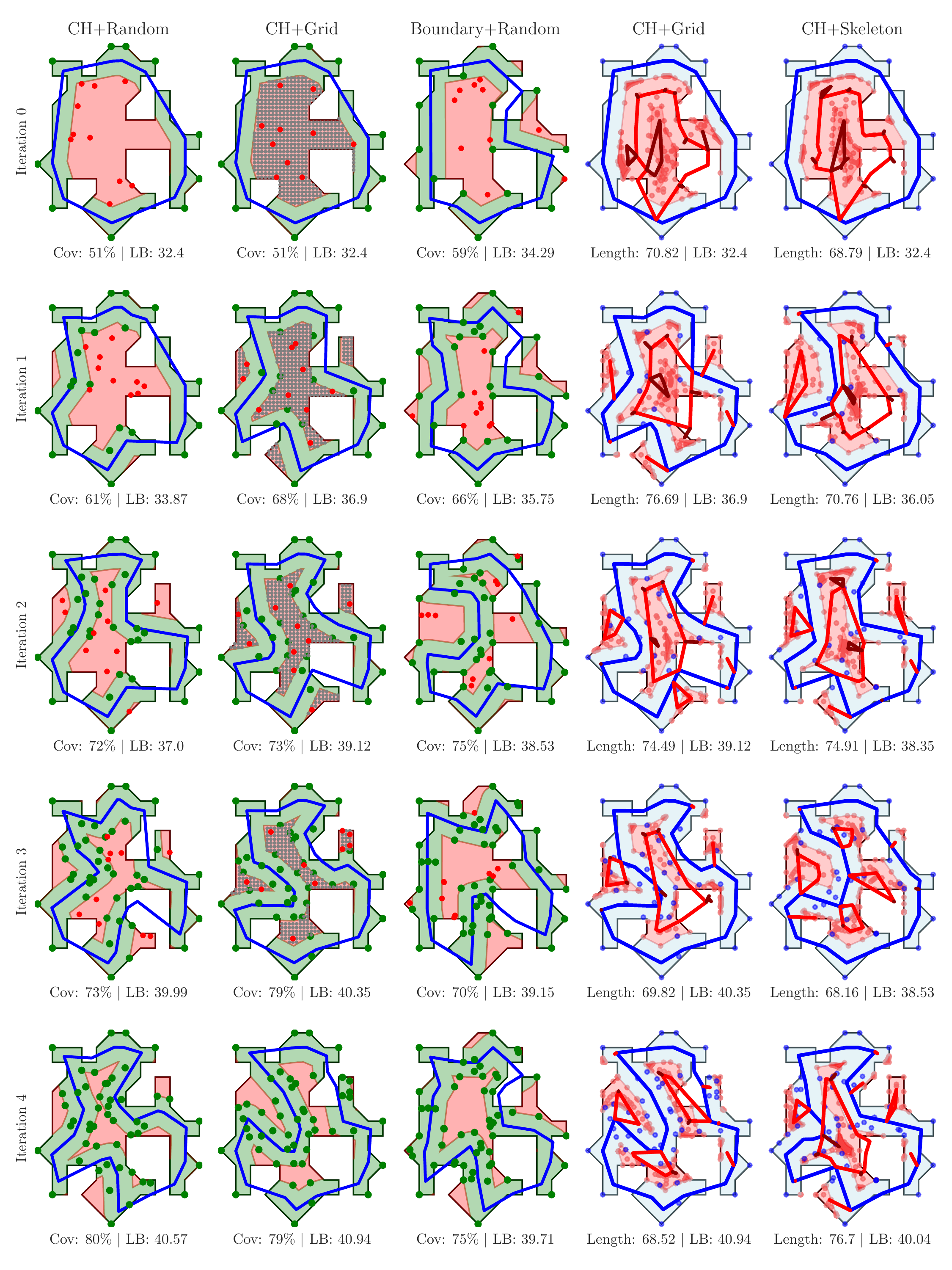}
  \caption{Examples of different witness placement strategies.
  The examples show just the first 5 iterations of the lower bound on the left and of the upper bound on the right.
  The trajectory is highlighted in blue, the used witnesses in green.
  The lower bounds show the covered area in green (lower bound) resp.\ blue (upper bound) and missed area in red.
  The light red tours for the upper bounds show the tour extensions after the first iteration, and the dark red tours later extensions.
  The new witnesses for the next iteration are highlighted in red, while the corresponding candidate set is visualized by gray points.
  The best lower bound achieved in these 5 iterations is \num{40.94} and the best upper bound \num{68.16}.}
  \label{fig:eval:witnessexamples:ubexamples}
\end{figure*}

%% file: 06-conclusion.tex
\section{Conclusion}
The Lawn Mowing Problem generalizes a number of notoriously difficult geometric
optimization problems, making it both important and highly challenging.
We have provided progress on several aspects of the LMP,
both on the theoretical and the practical side, demonstrating (for the first
time ever) that practically useful progress may be achievable.

A spectrum of open problems remain. 
On the theoretical side, we already stated \cref{prob1} and \cref{prob2}.
On the practical side, two critical questions are the development
of additional approaches for speeding up the computation
of lower bounds, as well as more efficient ways to compute upper bounds, 
along with additional local improvement heuristics.
In addition, the concept of $\varepsilon$-robustness deserves further exploration.

%% file: 04C-cetsp-approach.tex
\onecolumn
\section{Pseudocode for the algorithm}

\renewcommand{\algorithmicrequire}{\textbf{Input:}}
\renewcommand{\algorithmicensure}{\textbf{Output:}}

  \begin{algorithm}
  \caption{Pseudocode for the primal-dual algorithm.}
  \label{algo:cetsp-feasible}

  \begin{algorithmic}
  \Require{A polygon $\Pol$, a cutter $C$ with radius $r$, max iterations $max_i$}
  \Ensure{A feasible tour $\T$ and a lower bound $b_l$}
    \State $b_l = 0$
    \State $\T_{best}= \emptyset$
    \State Initialize $W_0$ with some points

    \For{$i\gets0, 1, \dots, max_i-1$}
      \If{$\T_{best} = \emptyset \hspace{1mm} \lor \hspace{1mm} \ell(\T_{best}) > b_l$}
        \State Solve CETSP for centers $W_i$ and $r_k=r$ $(0\leq k < |W_i|)$ to get $\T_i$ and a lower bound $b^{CETSP}_l$
        \State $b_l = \max \{b^{CETSP}_l, b_l\}$
        \State $W_{i+1} = W_i$

        \Do
          \State $uncoveredRegions = \Pol \setminus (\T_i \oplus C)$
          \ForAll{$R_j \in uncoveredRegions$}
              \State Initialize $W_{R_j}$ from $R_j$
              \State $p = \argmin_{p \in \partial\T_i} d(p, R_j)$
              \State Solve CETSP for centers $W_{R_j} \cup \{p\}$ and $r_k=r$,$r_p=0$ $(0\leq k < |W_{R_j}|)$ to get \State $\T_{R_j}$
              \State $W_{i+1} = W_{i+1} \cupdot W_{R_j}$
              \State Connect $\T_i$ and $\T_{R_j}$ via $p$
          \EndFor
        \doWhile{$uncoveredRegions\neq\emptyset$}

        \If{$\T_{best} = \emptyset$}
          \State $\T_{best} = \T_i$
        \Else
          \State $\T_{best} = \argmin_{T \in \{\T_i, \T_{best}\}} \ell(T)$
        \EndIf
      \EndIf
    \EndFor
    \State \Return {$\T_{best}$, $b_l$}
    \end{algorithmic}
  \end{algorithm}

%% file: 05A-additional-figures.tex
\onecolumn
\section{Illustrations of instances and solution progress}

\begin{figure*}[h!]
  \centering
  \includegraphics[width=.8\linewidth]{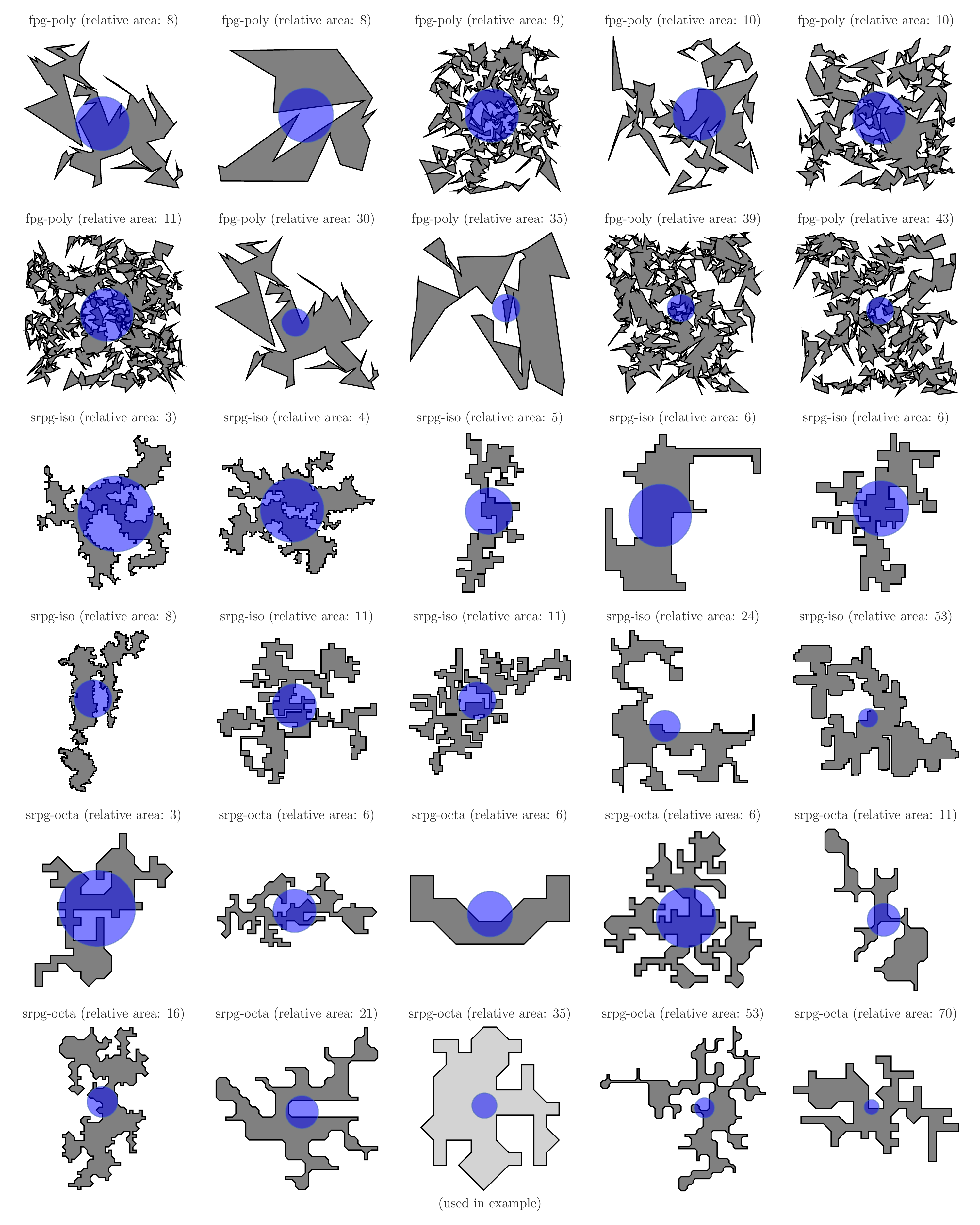}
  \caption{Examples of the used polygons. The blue circles show the respective tool size; the
\emph{relative area} denotes the ratio of convex hull area and cutter area.}
\label{fig:benchmarks}
\end{figure*}


\begin{figure*}
  \centering
  \includegraphics[width=\columnwidth]{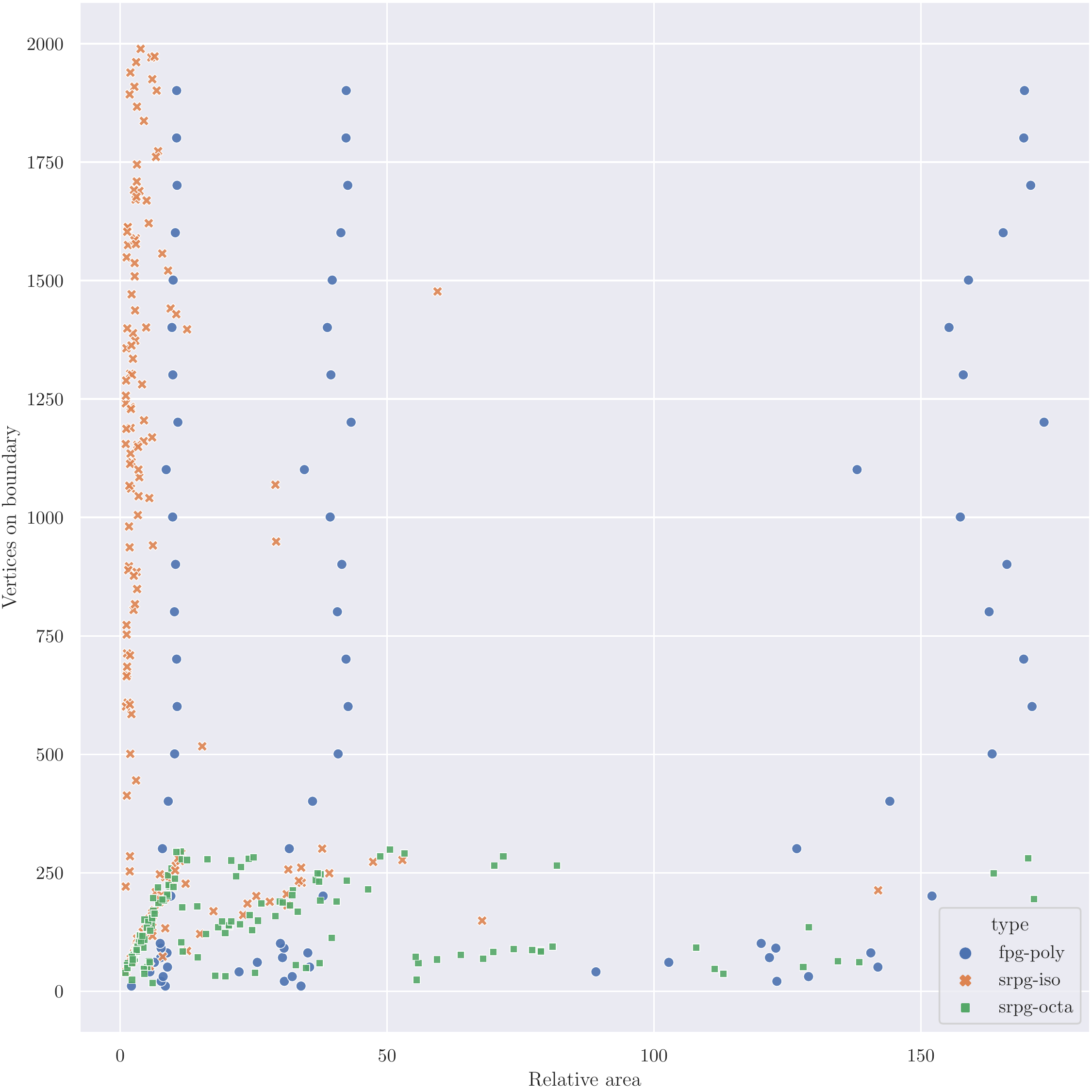}
  \caption{Enlarged distribution of instances, subdivided into different types.}
  \label{fig:distribution-enlarged}
\end{figure*}